\documentclass[10pt,aps,prl,twocolumn,superscriptaddress]{revtex4-2}
\usepackage[english]{babel}
\usepackage{microtype}
\usepackage{ragged2e}
\usepackage[autostyle=true]{csquotes} 
\usepackage{fontawesome5}

\usepackage{orcidlink}
\usepackage{amsmath}
\usepackage{mathtools}
\usepackage{amsfonts}
\usepackage{amssymb}
\usepackage{dsfont}
\usepackage{nicefrac}
\usepackage{braket}
\usepackage{diffcoeff}
\usepackage{cancel}

\usepackage{zref-clever}
\zcsetup{
    cap=true, 
    nameinlink=false, 
}

\zcRefTypeSetup{equation}{
    abbrev=true, 
}

\zcRefTypeSetup{figure}{
    abbrev=true,
}

\newcommand{\cref}[1]{\zcref{#1}}
\newcommand{\Cref}[1]{\zcref[S]{#1}}

\usepackage{graphicx} 
\graphicspath{{./figures/}}

\makeatletter
\newcounter{panel}[figure]

\def\p@panel{\thefigure}

\makeatother

\usepackage[shortlabels]{enumitem}

\usepackage{keytheorems} 
\newkeytheorem{theorem}[
]
\newkeytheorem{proposition,lemma,corollary}[
    style=plain,
    sibling=theorem 
]
\newkeytheorem{definition}[
    style=definition,
    sibling=theorem
]
\newkeytheorem{remark,example}[
    style=remark,
    sibling=theorem
]
\newkeytheorem{hypothesis}[
    style=plain,
    sibling=theorem,
    name=Hypothesis,
    refname=hypothesis,
    Refname=Hypothesis,
]

\usepackage{algorithm}
\usepackage{algpseudocodex}

\usepackage{booktabs}
\usepackage{multirow}

\usepackage{appendix}

\usepackage{xcolor}
\definecolor{royalpurple}{rgb}{0.47, 0.32, 0.66}
\definecolor{dblue}  {RGB}{20,66,129}
\definecolor{ddblue} {RGB}{11,36,69}
\definecolor{jblue}  {RGB}{20,50,100}
\definecolor{cblue}{rgb}{0.16, 0.32, 0.75}
\definecolor{cred}{rgb}{0.7, 0.11, 0.11}

\usepackage{hyperref} 
\hypersetup{%
	colorlinks = true,
	linkcolor=cred,
	citecolor=cblue,
	urlcolor=black,
}

\usepackage{tikz}
\usetikzlibrary{babel}
\usetikzlibrary{decorations.pathreplacing,calc,shapes.geometric}
\usetikzlibrary{arrows.meta, positioning}
\usetikzlibrary{shapes.multipart}

\usepackage{comment}

\newcommand{\id}{\mathds{1}}

\DeclareMathOperator{\Imag}{Im}
\DeclareMathOperator{\mspan}{Span}
\DeclareMathOperator{\Ker}{Ker}
\DeclareMathOperator{\Ran}{Ran}

\DeclareMathOperator{\Mat}{Mat}

\DeclarePairedDelimiterX{\norm}[1]\lVert\rVert{
	\ifblank{#1}{\:\cdot\:}{#1}
}
\DeclarePairedDelimiterX{\opnorm}[1]\lVert\rVert{
	\ifblank{#1}{\:\cdot\:}{#1}
}

\newcommand{\iu}{\mathrm{i}\mkern1mu}
\newcommand{\e}{\mathrm{e}}

\newcommand{\nnum}{\mathbb{N}}
\newcommand{\rnum}{\mathbb{R}}
\newcommand{\cnum}{\mathbb{C}}

\newcommand{\hilbert}{\mathcal{H}}
\newcommand{\domain}{\mathcal{D}}
\newcommand{\schwartz}{\mathcal{S}}

\newcommand{\Jop}{\mathcal{J}}

\newcommand{\dspace}{\mathcal{N}}
\newcommand{\sphere}{\mathcal{B}}
\newcommand{\matupper}{\mathbb{U}}
\newcommand{\wronsk}{\mathcal{W}}

\newcommand{\esssa}{\textcolor{green!55!black}{\faCheck}}
\newcommand{\hasext}{\textcolor{orange!90!white}{\faInfoCircle}}
\newcommand{\noext}{\textcolor{red!75!black}{\faTimes}}

\begin{abstract}
    Hermitian polynomials in the quadrature operators $x$ and $p$ are the basic ingredients of continuous-variable quantum information, and increasingly exotic ones are now routinely
engineered in superconducting circuits and trapped ions. However, such a polynomial need not
define unique quantum dynamics: it may define exactly one, none, or an
infinite family of physically distinct ones. Such obstructions are invisible to algebraic arguments and can compromise finite-dimensional truncations, which may fail to converge or silently converge to the wrong dynamics. 
We develop a method, implemented
as an algorithm, which decides which case occurs---and how large the family
is---from the coefficients alone, with no functional-analytic input from the
user. For Gaussian operators the answer is always one; for non-Gaussian ones, a regime continuous-variable quantum computing cannot do without, the other outcomes occur already for simple and physically natural Hamiltonians. We apply the algorithm to operators
drawn from the recent literature.
\end{abstract}

\begin{document}

\title{When is a bosonic operator a Hamiltonian?}
\date{\today}

\author{Felix Fischer\,\orcidlink{0009-0001-3040-8184}}
\affiliation{Department of Physics, Friedrich-Alexander-Universität Erlangen-Nürnberg, Staudtstraße 7/B3, 91058 Erlangen, Germany}

\author{Hermann Schulz-Baldes\,\orcidlink{0000-0003-0304-4140}}
\affiliation{Department of Mathematics, Friedrich-Alexander-Universität Erlangen-Nürnberg, Cauerstraße 11, 91058 Erlangen, Germany}

\author{Daniel Burgarth\,\orcidlink{0000-0003-4063-1264}}
\affiliation{Department of Physics, Friedrich-Alexander-Universität Erlangen-Nürnberg, Staudtstraße 7/B3, 91058 Erlangen, Germany}

\author{Davide Lonigro\,\orcidlink{0000-0002-0792-8122}}
\affiliation{Department of Physics, Friedrich-Alexander-Universität Erlangen-Nürnberg, Staudtstraße 7/B3, 91058 Erlangen, Germany}

\maketitle

\begin{table}[t!]
    \centering
    \begin{tabular*}{\columnwidth}{@{\extracolsep{\fill}} p{5.7cm} p{1.7cm} c @{}}
    \toprule
    Operator & Used in: & Result \\
    \midrule
    $(a^\dagger)^3 + a^3$  & \multirow{2}{2.2cm}{\cite{bazavan-squeezingtrisqueezingquadsqueezing-2026,eriksson-universalcontrolbosonic-2024,deng-generatingcharacterizinggeneralized-2026,ashhab-propertiesdynamicsgeneralized-2025,menard-emissionphotonmultiplets-2022,chang-observationthreephotonspontaneous-2020,sutherland-universalhybridquantum-2021,saner-generatingarbitrarysuperpositions-2026,bu-mechanicaltrisqueezingcircuit-2026}} & \hasext \\
     $(a^\dagger)^4 + a^4$ & & \hasext \\
     $(a^\dagger)^2 a + a^\dagger a^2$ & \cite{leibfried-trappedionquantumsimulator-2002} & \hasext \\
     \midrule
     $\lambda (a^\dagger)^2 a^2$+$\eta((a^\dagger)^2+a^2)$  & \cite{grimm-stabilizationoperationkerrcat-2020,ruiz-twophotondrivenkerr-2023,karthik-steadystatespectralkissing-2026,hajr-highcoherencekerrcatqubit-2024,leghtas-confiningstatelight-2015,puri-stabilizedcatdriven-2019,baskov-exactamplitudesparametric-2025,frattini-observationpairwiselevel-2024} & \esssa \\
     $\lambda (a^\dagger)^2 a^2 + \kappa a^\dagger a+\eta((a^\dagger)^2+a^2)  $ & \cite{sivak-kerrfreethreewavemixing-2019,chavez-carlos-spectralkissingits-2023,reynoso-quantumtunnelinglevel-2023,marti-quantumsqueezingnonlinear-2024} & \esssa \\
      $\lambda (a^\dagger)^2 a^2+ \kappa a^\dagger a + \eta_1 a + \eta_2 a^2 + \eta_3 a^\dagger a^2+h.c. $ & \cite{yu-stabilizingkerrarbitrary-2024,puri-stabilizedcatdriven-2019,puri-engineeringquantumstates-2017} & \esssa \\
     \midrule
     $\lambda x^3$ & \cite{park-deterministicnonlinearphase-2018,kasri-nonlinearphasegates-2026,rakhubovsky-stroboscopichighordernonlinearity-2021,eriksson-universalcontrolbosonic-2024,junsheng-cubicphasegate-2026,fontaine-programminganharmonicpotentials-2026,hillmann-universalgateset-2020} & \esssa \\
     $\kappa a^\dagger a + \mu_1 x + \mu_2 x^2 + \mu_3 x^3 + \lambda x^4$, $\lambda>0$ & \cite{setter-characterizationnonlinearitiesmechanical-2019,zhu-classicalsymptomsquantum-1993,hillmann-universalgateset-2020} & \esssa \\
     $\kappa a^\dagger a + \mu_1 x + \mu_2 x^2 + \mu_3 x^3 - \lambda x^4$, $\lambda>0$ & \cite{zhu-classicalsymptomsquantum-1993,hillmann-universalgateset-2020,lu-resolvingfockstates-2023} & \hasext \\
     $p^2 + x^2 + \lambda x^3$ & \cite{rakhubovsky-stroboscopichighordernonlinearity-2021} & \hasext \\
     \midrule
   $\iu [(a^\dagger a)^2,x]$ & \multirow{4}{*}{\cite{lloyd-quantumcomputationcontinuous-1999}} & \hasext  \\
         $\iu [(a^\dagger a)^2,p]$ & & \hasext \\
         $[x,[(a^\dagger a)^2,\iu(a^\dagger)^2-\iu a^2]]$ & & \esssa \\
         $[p,[(a^\dagger a)^2,\iu(a^\dagger)^2-\iu a^2]]$ & & \esssa \\ 
    \midrule
         $x^2 p +xpx+px^2$ & \multirow{5}{*}{\cite{gordillo-hachuel-quantummetrologicaladvantage-2026,guo-metrologicalsensitivitygaussian-2025}} & \hasext \\ 
         $xp^2 + pxp + p^2 x$ & & \hasext \\ 
         $x^3 p +x^2px+xpx^2+px^3$ & & \noext \\ 
         $x p^3 +pxp^2+p^2 x p + p^3 x$ & & \noext \\ 
         $x^2 p^2 + xpxp + xp^2 x + px^2p+pxpx+p^2 x^2$ & & \hasext \\ 
    \midrule
     $\kappa_1 a^\dagger a-\frac{\kappa_2}{2} (a^\dagger)^2 a^2 + \frac{\eta}{2} ((a^\dagger)^3 + a^3) + \lambda((a^\dagger)^4 a + a^\dagger a^4)$ & \cite{kwon-realisationprotectedcatqutrit-2026} & \hasext \\
    \midrule
     $\eta_1 (a^\dagger)^2 a^2+ \eta_2 (a^\dagger a^3 + (a^\dagger)^3 a) +a^4 + (a^\dagger)^4$ & \multicolumn{2}{l}{\cite{boutin-effecthigherordernonlinearities-2017} \, See \Cref{fig:sweep}}  \\ 
     \midrule
    $\iu \lambda(a-a^\dagger) -\iu \frac{\lambda^3}{2}((a^\dagger)^2 a-a^\dagger a^2)$ & \multirow{3}{*}{\cite{kasri-nonlinearphasegates-2026}} & \hasext \\
    $\frac{\lambda^2}{2!}((a^\dagger)^2+a^2) -\frac{\lambda^4}{3!}((a^\dagger)^3 a + a^\dagger a^3)$ & & \hasext \\
    $-\iu \frac{\lambda^3}{3!}((a^\dagger)^3 - a^3)- \iu\frac{\lambda^5}{5!} ((a^\dagger)^4 a-a^\dagger a^4)$ & & \hasext \\
     \bottomrule
    \end{tabular*}
    \caption{Result of the algorithm for a selection of operators appearing in the
literature. \esssa~indicates that the operator defines a unique physical system;
\hasext~that additional physical input is required to fix one; and \noext~that
no physical system can have the operator as its Hamiltonian. Here $a$ and $a^\dagger$ are the bosonic annihilation and creation
operators;
all remaining symbols are real constants with $\lambda \neq 0$ throughout. 
A parameter sweep of the operator of~\cite{boutin-effecthigherordernonlinearities-2017} is shown in \cref{fig:sweep}.
The references
listed indicate works in which the operator appears; they
carry no implication as to how it is treated there.}
    \label{tab:literature-results}
\end{table}
\emph{Introduction}---Non-Gaussian operations generated by polynomials of degree higher
than two in the
position and momentum operators $x$ and $p$ (equivalently, in the annihilation
and creation operators $a=x+\iu p$ and $a^\dagger=x-\iu p$) are
strictly necessary for
continuous-variable quantum
computing~\cite{lloyd-quantumcomputationcontinuous-1999,hillmann-universalgateset-2020,arzani-effectivedescriptionsbosonic-2025,chabaud-energybosonscomputational-2025,chabaud-bosonicquantumcomputational-2026,rudolph-physicaluniversalmodel-2026},
and rapid experimental progress has made increasingly exotic Hamiltonians
available in bosonic platforms such as superconducting
circuits~\cite{chang-observationthreephotonspontaneous-2020,eriksson-universalcontrolbosonic-2024,deng-generatingcharacterizinggeneralized-2026,kwon-realisationprotectedcatqutrit-2026,saner-generatingarbitrarysuperpositions-2026,fontaine-programminganharmonicpotentials-2026}
and trapped
ions~\cite{bazavan-squeezingtrisqueezingquadsqueezing-2026,mcgarry-programmablequantumsimulation-2026,junsheng-cubicphasegate-2026}.
This marks a departure from most of the history of quantum optics, whose experimentally relevant regime was that of Gaussian operations—polynomials in $x$ and $p$ of degree at most two. In that regime, a Hermitian generator always defines a unique dynamics~\cite{weedbrook-gaussianquantuminformation-2012,reed-mmmp2-fourier-1975,folland-harmonicanalysisphase-2016};
beyond it, this can no longer be taken for granted, and the first and unavoidable step towards understanding what non-Gaussian operators can and cannot do is to answer the following question: \textit{when does a
non-Gaussian bosonic operator specify a unique unitary dynamics?}

The answer is nontrivial even for some of the simplest non-Gaussian operators. The
higher-order squeezing operators $(a^\dagger)^k+a^k$, $k\geq3$, were recently
shown to define not one dynamics but a whole manifold of physically distinct
ones~\cite{fischer-selfadjointrealizationshigherorder-2026}, settling a
long-standing debate on their
physicality~\cite{fisher-impossibilitynaivelygeneralizing-1984,braunstein-generalizedsqueezing-1987,nagel-higherpowersqueezed-1997,gardas-multiphotonrabimodel-2013,lo-commentmultiphotonrabi-2014,gardas-replycommentmultiphoton-2014,ashhab-propertiesdynamicsgeneralized-2025,gordillo-hachuel-commentpropertiesdynamics-2026,ashhab-replycommentproperties-2026};
analogous results hold for multiphoton light--matter Hamiltonians such as the
$k$-photon Rabi
model~\cite{braak-kphotonquantumrabi-2025,fischer-selfadjointextensionskphoton-2026}.
Such operators cannot be used as Hamiltonians without additional physical input,
and simulating them is treacherous: their finite-dimensional truncations need not
converge~\cite{ashhab-finitedimensionalapproximationsgeneralized-2026}, and when
they do, the limit may be a dynamics selected by the discretization rather than
by the physics~\cite{wrongbox}, with nothing in the output to indicate it. With
non-Gaussian operators proliferating in the literature, such pathologies must be
identified systematically rather than case by case---yet every case settled so
far has required expert-level functional analysis, out of reach of most
practitioners.

This Letter closes this gap with a numerical method, based on matrix-valued Weyl theory~\cite{schulz-baldes-geometryweyltheory-2010}, that takes as input nothing but the coefficients of any Hermitian polynomial in $x$ and $p$ and returns two integers $(n_+,n_-)$—the deficiency indices of the operator. These determine whether the polynomial defines a single physical dynamics, none at all, or infinitely many physically distinct ones.
The truncation schemes routinely employed in quantum simulation~\cite{qutip5}
cannot answer this question, let alone produce the deficiency indices.
No prior knowledge of functional analysis is required of the user, and wherever an
analytic answer is known we have verified that the algorithm reproduces it.
Beyond those cases, the method determines the deficiency indices of many
operators for which no analytic treatment is currently available. In particular, we apply it to non-Gaussian operators appearing in recent experimental
and theoretical work on continuous-variable quantum computing, collected in
\Cref{tab:literature-results}: the algorithm finds that a substantial fraction of them either admit several dynamics or none at all.
Finally, the method can analyze families of bosonic operators: scanning the coefficients of a given polynomial yields phase
diagrams of the deficiency indices, such as the one in \cref{fig:sweep}.

\begin{figure}[b]
    \centering
    \resizebox{\columnwidth}{!}{%
    \begin{tikzpicture}[
    font=\scriptsize,
    box/.style={rectangle, draw, rounded corners=4pt, text width=4.1cm,
                inner sep=3.5pt, align=center, line width=0.7pt},
    dia/.style={diamond, draw, aspect=1.5, inner sep=1pt, align=center,
                rounded corners=2pt, line width=0.7pt},
    arrow/.style={-Latex, line width=0.7pt},
    lbl/.style={font=\scriptsize, anchor=east, inner sep=3pt}
]
\def\spine{-3.0}
\def\mid{-1.2}

\node[box] (op) at (\mid,0) {Hermitian operator $H$};
\node[dia] (q)  at (\mid,-1.3) {Deficiency\\ indices $n_\pm$};

\node[box, draw=green!55!black, fill=green!12] (one) at (0,-2.9)
    {Unique physical realization};
\node[box, draw=orange!90!black, fill=orange!18] (multi) at (0,-4.3)
    {Additional information about the system required};
\node[box, draw=red!75!black, fill=red!12] (none) at (0,-5.8)
    {No physical realization possible};

\draw[arrow] (op) -- (q);
\draw[line width=0.7pt] (q.south) -- (\mid,-2.1) -- (\spine,-2.1) -- (\spine,-5.8);
\draw[arrow] (\spine,-2.9) -- (one.west);
\draw[arrow] (\spine,-4.3) -- (multi.west);
\draw[arrow] (\spine,-5.8) -- (none.west);

\node[lbl] at (\spine,-2.9) {$n_+=n_-=0$};
\node[lbl] at (\spine,-4.3) {$n_+=n_-\neq0$};
\node[lbl] at (\spine,-5.8) {$n_+\neq n_-$};
\end{tikzpicture}
    }
    \caption{The three possible physical realizations of a Hermitian
operator $H$, classified by its deficiency indices $n_\pm$. If
exactly one unitary evolution is compatible with $H$, the operator alone fixes
the physics. If several are, further physical input is needed to single one out, and different choices
describe genuinely different systems. If none is, $H$ is not the Hamiltonian of
any physical system.}
    \label{fig:self-adjoint}
\end{figure}
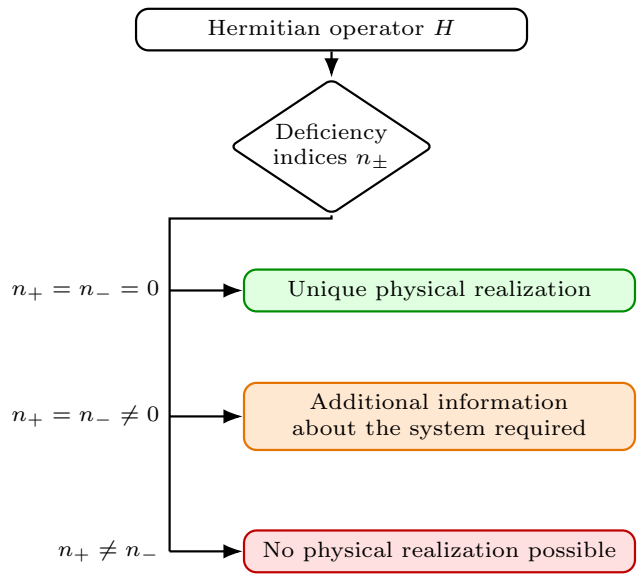

\medskip\emph{What is a Hamiltonian?}---In applications, one requires only that $H$ be Hermitian, or symmetric in the mathematical terminology. But Hermiticity is not enough to qualify an operator as a physical Hamiltonian. So what is?
Since the answer determines how the verdicts of our algorithm are to be read, we briefly recall it along a route slightly different from, but equivalent to, the standard one~\cite{reed-mmmp1-funkana-1980,schmudgen-unboundedselfadjointoperators-2012}: rather than postulating $H$ and
asking what dynamics it generates, let us take the converse route and start from
the dynamics.

The states of a closed quantum system are vectors $\psi$ in a Hilbert space
$\hilbert$, evolving as $\psi(t)=U(t)\psi(0)$ with $U(t)$ unitary and strongly
continuous. The Hamiltonian is then \emph{defined} as the generator of this
evolution, $-\iu H\psi=\lim_{t\to0}\left(U(t)\psi-\psi\right)/t$. In infinite
dimensions this limit need not exist for every $\psi$, so $H$ comes with a
domain $\domain(H)$---the states on which it is defined---which, whenever $H$ is unbounded, is never the whole of $\hilbert$. This
is always the situation in continuous-variable systems, where $x$ and $p$ are
themselves unbounded. 

A short calculation shows that such an $H$ is necessarily
Hermitian. It also satisfies a third, less familiar property~\cite{gustafson-mathematicalconceptsquantum-2011}: the equations
\begin{equation}
    \label{eq:deficiency}
    (H\pm\iu)\psi=\eta,
\end{equation}
with $\psi\in\domain(H)$, can be solved for every $\eta\in\hilbert$, as one sees by writing the solution
explicitly as a Laplace transform of $U(t)$~\cite{schmudgen-unboundedselfadjointoperators-2012,reed-mmmp2-fourier-1975}. Remarkably, the converse
holds as well: by Stone's theorem, 
a densely defined closed Hermitian operator for which, in addition,~\cref{eq:deficiency}
is always solvable generates a unique unitary evolution~\cite{schmudgen-unboundedselfadjointoperators-2012,reed-mmmp2-fourier-1975}. Such operators
are called self-adjoint, and they are exactly the physical Hamiltonians.

The failure of~\cref{eq:deficiency} is measured by two integers, the deficiency
indices $n_\pm=\dim\Ran(H\pm\iu)^\perp$ of $H$, which count the independent states that cannot be
reached~\cite{schmudgen-unboundedselfadjointoperators-2012,reed-mmmp2-fourier-1975}. Their meaning is physical: they count the independent ways in which the system
can escape---to a boundary, to a singularity, or, for the bosonic operators of
interest here, to arbitrarily high occupation numbers---in finite time, and
therefore the amount of information one must provide about what happens when it
does; for these operators, this is made precise through the
probability flux at infinity, see \cref{sec:flux}. Von Neumann's
theorem~\cite{reed-mmmp2-fourier-1975,schmudgen-unboundedselfadjointoperators-2012} makes the three possibilities precise (see \cref{fig:self-adjoint}). If $n_+=n_-=0$, nothing escapes: the probability flux at infinity vanishes identically, and $H$ fixes a unique evolution. If $n_+=n_-\neq0$, the evolution is fixed only after further physical input is supplied: one must prescribe how the flux at infinity is forced to vanish—in effect, a boundary condition at infinity. The admissible boundary conditions are in one-to-one correspondence with the $n_+\times n_+$ unitary matrices 
describing how the escape channels are scattered into one another at infinity: unitarity is precisely the condition that whatever leaves through one channel returns through another. The operator thus determines not one dynamics but a manifold of them.
If $n_+\neq n_-$, there are more independent ways to escape in one direction of time than in the other, so no boundary condition at infinity can make the flux vanish, and no unitary evolution exists at all: the best $H$ can generate is a one-sided semigroup, isometric in one time direction while probability drains away in the other, and never invertible. Such an $H$ describes an irreducibly dissipative process, and is not the Hamiltonian of any closed system. Each case is encountered in practice, as our results in \cref{tab:literature-results} show.

In practice, of course, physics proceeds in the opposite direction: one writes
down a candidate $H$ on a convenient dense domain, and the question is which of
the three cases applies. Hermiticity is immediate by inspection. The deficiency
indices are not, and this is where our algorithm comes in.

\begin{figure}
\centering
\includegraphics[width=\columnwidth]{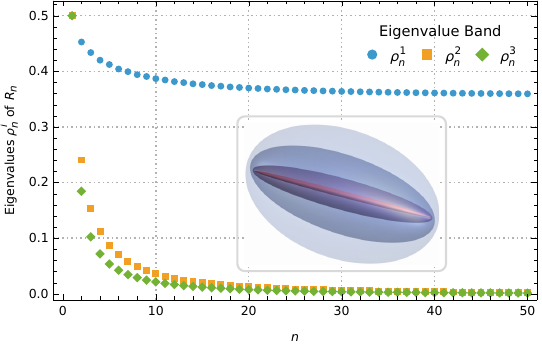}
\caption{The algorithm at work on the operator $x^2p+xpx+px^2$
of~\cite{gordillo-hachuel-quantummetrologicaladvantage-2026,guo-metrologicalsensitivitygaussian-2025}. Eigenvalues $\rho^1_n \geq \rho^2_n \geq \rho^3_n$ of the
radius matrices $R^+_n$ against $n$: two of the three decay to zero while one
remains finite, giving $n_+=1$. Equivalently (inset), the Weyl disks $\sphere^+_n$
determined by $R^+_n$---shown for $n=4,10,40,200$---are nested one inside the
next and collapse from three dimensions onto a one-dimensional segment, whose
dimension is $n_+$. The operator therefore admits a whole family of distinct
unitary evolutions rather than a single one. The same results are obtained for $n_-$. Counting the surviving eigenvalues
is the only numerical step in the algorithm.}
\label{fig:nested}
\end{figure}

\medskip\emph{The algorithm}---We now return to the operators of interest and describe how the algorithm computes their deficiency indices. A single-mode bosonic operator is given by a
normal-ordered polynomial
\begin{equation}
    \label{eq:polynomial}
    p_M=\sum_{k,l=0}^{s}M_{k,l}\,(a^\dagger)^k a^l
\end{equation}
in the creation and annihilation operators, Hermitian whenever the coefficient
matrix $M$ is. Via $x=\tfrac{1}{2}(a+a^\dagger)$ and
$p=\tfrac{\iu}{2}(a^\dagger-a)$, every polynomial in position and momentum takes
this form. In formal manipulations one tacitly works on states for which the
usual operations with $a$ and $a^\dagger$---applying them repeatedly, commuting
them, differentiating---may be performed without concern.
Making this precise
means fixing an initial domain: whenever none is specified,
as is almost always the case, one of the standard choices has in effect been
made. These are the Schwartz space $\mathcal{S}(\rnum)$,
the smooth compactly supported functions $C^\infty_0(\rnum)$, and the finite linear combinations of Fock
states. The choice turns out not to matter: we show in \cref{sec:domains},
\Cref{lem:all-domains}, that
$p_M$ has the same closure, and thus the same deficiency indices, on all three.

Let us sketch how the algorithm works, leaving the details to~\cref{sec:algorithm}. In the number basis, a Hermitian polynomial $p_M$ acts as a
linear recurrence---a block Jacobi operator, as we show
in~\cref{sec:reduction}---and the algorithm proceeds by tracking how the solutions of
that recurrence are constrained as more and more excitations are included. It
first computes a sequence of so-called radius matrices $R_n^+$ and $R_n^-$ (\cref{sec:weyl-disks}),
strictly positive definite matrices acting on $\cnum^d$, where $d$ is the
maximum shift in excitation number induced by one application of $p_M$. Each of them specifies the size of a Weyl disk $\sphere_n^\pm\subset\cnum^{d\times d}$~\cite{schulz-baldes-geometryweyltheory-2010}, which has a transparent meaning: if the recurrence is truncated after $n$ steps and an
arbitrary boundary condition is imposed at the cut, $\sphere_n^\pm$ is the set
of values that the Green function of the truncated operator can take (\cref{prop:weyl-sphere}). 
These disks are strictly nested, $\sphere^\pm_{n+1}\subset\sphere^\pm_n$ for every $n$, so the limits $\sphere^\pm=\lim_{n\to\infty}\sphere^\pm_n$ and $R^\pm=\lim_{n\to\infty}R^\pm_n$ always exist—but the latter need not retain full rank $d$, in which case the disks collapse onto lower-dimensional ones. The deficiency indices $n_\pm$ of $p_M$ are precisely the ranks of $R^\pm$, equivalently the dimensions of $\sphere^\pm$ (\cref{cor:main}).
In particular, $p_M$ is associated with a unique time evolution if and only if
both $\sphere^\pm$ consist of a single point; that is, if and only if the boundary
condition at the cut ceases to matter as $n\to\infty$ (Weyl's limit-point case). 
The case $n_+\neq n_-$, in which no unitary dynamics exists,
cannot occur
in two situations of practical importance: when all coefficients $M_{k,l}$
are real, since $p_M$ then commutes with complex conjugation in the number
basis~\cite{reed-mmmp2-fourier-1975}; and when $p_M$ is bounded from below, since a
semibounded Hermitian operator always admits a self-adjoint extension, the Friedrichs
one~\cite{reed-mmmp2-fourier-1975,schmudgen-unboundedselfadjointoperators-2012}.

In practice it is enough to compute the positive eigenvalues of the radius matrices and
examine their asymptotics, the number of eigenvalues with nonzero limit being
the dimension of $\sphere^\pm$. \Cref{fig:nested} shows both pictures for
the operator $x^2p+xpx+px^2$
of~\cite{gordillo-hachuel-quantummetrologicaladvantage-2026,guo-metrologicalsensitivitygaussian-2025}:
the 
three-dimensional disks $\sphere^+_n$ collapse onto a segment, and
correspondingly a single eigenvalue remains nonzero, giving $n_+=1$. We further
describe in~\cref{sec:convergence,sec:raabe} how to stabilize the convergence and automate the
vanishing check, which is what makes the algorithm practical for sweeps over
families of Hermitian operators of the form \eqref{eq:polynomial}.

\medskip\emph{The algorithm at work}---We have implemented the algorithm in
Mathematica~\cite{Mathematica}; the package is publicly
available~\cite{notebook}, and we also provide a public web
interface~\cite{website}. Let us first use it to audit the literature on
continuous-variable quantum computing.
\Cref{tab:literature-results} collects Hermitian
operators appearing in experimental and theoretical work, where they serve as generators of gates or as
observables---roles that in either case require a unique associated unitary
group. Many of them do not have one: several admit multiple compatible evolutions, and some admit none at all. A
verdict of non-uniqueness is not a defect in the work in which the operator
appears; it marks a place where the physical specification is incomplete, and
where the missing input remains to be provided.

Finally, we demonstrate how the algorithm can be used to analyze entire families
of operators. We consider the operator $(a^\dagger)^4+a^4+\eta_1\left((a^\dagger)^3a+a^\dagger a^3\right)+\eta_2(a^\dagger)^2a^2$ from \cite{boutin-effecthigherordernonlinearities-2017}, sweeping over the real parameters $\eta_1,\eta_2$, applying the algorithm at each
point, and plotting the resulting deficiency indices in \cref{fig:sweep}.
Pronounced thresholds between regimes are clearly visible: for some parameter
combinations the operator is associated with a unique unitary dynamics, while in
others additional information about the physical system is needed to fix a time
evolution. Operators lying close to the boundary of the former region are known
to yield dynamics highly sensitive to small parameter
variations~\cite{fischer-essentiallysingularlimits-2026}. Such sweeps thus
quickly locate the regions worth subsequent analytical investigation.

\begin{figure}[t]
    \centering
    \includegraphics[width=.95\columnwidth]{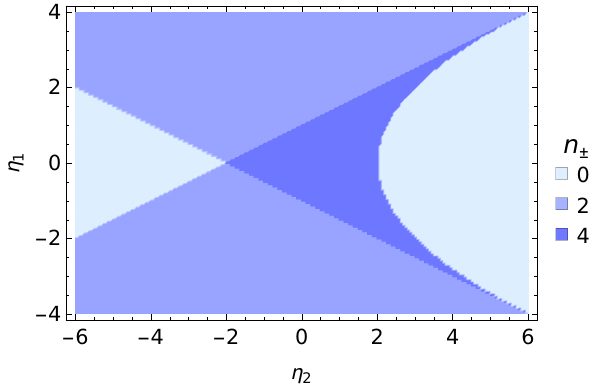}
    \caption{
  Deficiency indices $n_\pm$ of the operator
$(a^\dagger)^4 +a^4+ \eta_1((a^\dagger)^3 a+a^\dagger a^3) +\eta_2 (a^\dagger)^2 a^2$ from \cite{boutin-effecthigherordernonlinearities-2017},
obtained by applying the algorithm at each point of the $(\eta_1,\eta_2)$ plane.
Three regimes with sharp boundaries appear: $n_\pm=0$, where the
operator alone fixes a unique dynamics; $n_\pm=2$ and $n_\pm=4$, where
physical input is needed to select one of a family of inequivalent
evolutions. A sweep of this kind takes a single run of the
algorithm per point and requires no analytical input. }
    \label{fig:sweep}
\end{figure}

\medskip
\emph{Discussion}---What this Letter offers the practitioner is a decision procedure where
previously there was only case-by-case analysis. If the algorithm returns that a
bosonic operator generates a unique time evolution, the standard toolkit---in
particular numerical truncation---can be used with confidence. If it returns that several evolutions exist, no amount of numerical effort will
single one out: the operator alone does not determine the physics, and
additional input is required before the model can be simulated or measured—and, as noted above, a truncation may converge to the wrong dynamics without warning~\cite{wrongbox}, so that the absence of any numerical warning is no evidence that the model is well posed. This is precisely the failure
the present check is meant to avoid. By telling where uniqueness fails and how large the family of admissible dynamics is, the algorithm singles out the operators for which a classification like that of the higher-order squeezing operators can now be sought, turning what has so far been a series of isolated results into a systematic program.

The one step at which numerics enters is the
mildest one available. 
Determining $n_\pm$ amounts, in principle, to counting the linearly independent square-summable solutions of a linear recurrence (\cref{prop:decomposition,lem:dspace}), which is
a delicate
computation even in the simplest cases, as shown for instance in \cite{fischer-selfadjointrealizationshigherorder-2026} for higher-order squeezing operators. Our method replaces that computation by a far more benign one. The eigenvalues
of the radius matrices are positive and monotonically decreasing, so a limit is
guaranteed to exist, and the entire numerical content of the algorithm is the
decision of whether that limit is zero. Deciding the vanishing of a monotone
sequence is about as stable a numerical question as one can ask, and it is
incomparably easier than reconstructing the deficiency indices from scratch.
Away from the boundaries in coefficient space at which the indices change value,
the resulting determination is therefore essentially certain; only in the
immediate vicinity of those boundaries does the distinction between a limit that
vanishes and one that is merely small become delicate. Comparison with the
analytic results available confirms this throughout (\cref{tab:literature-results}).

Several directions suggest themselves. While we have treated single-mode
bosonic operators with scalar coefficients here, the algorithm extends directly
to polynomials $p_M$ with matrix-valued coefficients. Every polynomial operator
coupling a single bosonic mode to a finite-dimensional matter system is of this
form, so the method covers a large class of light-matter systems, among them
the $k$-photon Rabi
model~\cite{braak-kphotonquantumrabi-2025,fischer-selfadjointextensionskphoton-2026}
and its generalizations.
The multimode case is harder. Self-adjointness of a polynomial operator on
several bosonic modes is in general
undecidable~\cite{chabaud-energybosonscomputational-2025}, so no single
algorithm can treat all such systems. In many cases of interest, however,
symmetry reduces a multimode operator to invariant subspaces on which it again
behaves as a single-mode operator, and our method applies there.

Beyond these questions of scope, the reduction to block Jacobi form has a
consequence of a different kind. Once a bosonic polynomial is recognized as a
block Jacobi operator, the whole machinery developed for such operators becomes
available. The mathematical results underlying our algorithm---the reduction of every
$p_M$ to block Jacobi form (\cref{prop:decomposition}), its independence of the
choice of domain (\cref{lem:all-domains}), and the resulting characterization
of the deficiency indices (\cref{cor:main})---can then be combined with the
spectral theory of these operators~\cite{swiderski-spectralpropertiesblock-2018,moszynski-barriernonsubordinacyabsolutely-2025,berg-orderindeterminatemoment-2014,budyka-deficiencyindicesblock-2024}
to establish the deficiency indices of specific families of bosonic polynomials
analytically.

We propose the algorithm as a routine first check wherever bosonic operators are employed, and in continuous-variable quantum computing in particular, as a
routine first check that places a model on firm ground before the standard
methods of physics and computer science---truncation, simulation, compilation
into gate sequences---are brought to bear on it with their convergence
guarantees intact. 
Deciding whether a bosonic operator is a Hamiltonian is now a computation that anyone can run.\medskip

\textit{Acknowledgments}---F.F. acknowledges support from the Munich Quantum Valley (MQV), which is supported by the Bavarian State Government with funds from the High-tech Agenda Bayern Plus.

\bibliography{refs-bibtex}

\appendix
 \makeatletter
\renewcommand{\p@subsection}{\thesection.}
\makeatother
\setcounter{secnumdepth}{3}
\onecolumngrid

\section{Block Jacobi operators and Weyl disks}\label{sec:weyl}
The algorithm in this work is based on the following 
observation: every operator $p_M$ considered here can be written as a
so-called block Jacobi operator.
For such operators, the theory of Weyl disks characterizes the
deficiency indices as the rank of the limit of a sequence of matrices, each constructed from the operator in finitely many steps---a
characterization amenable to numerical evaluation, as a closed formula would
not be.
This appendix develops that theory: \cref{sec:jacobi} introduces block
Jacobi operators, \cref{sec:counting} explains why a direct count of their
square-summable solutions is impractical, \cref{sec:weyl-disks} presents the
Weyl disk characterization of the deficiency indices, and \cref{sec:flux} gives
its interpretation in terms of a probability flux at infinity.
\Cref{sec:pm-domains} then closes the circle: there we establish the equivalence
and apply it to obtain the deficiency indices of $p_M$.
For an in-depth overview of block Jacobi operators we refer
to~\cite{berezanskii-expansioneigenfunctionsselfadjoint-1968,schulz-baldes-geometryweyltheory-2010,dyukarev-deficiencynumberssymmetric-2006,moszynski-blockjacobimatrices-2025}.
Many aspects of the theory are better known in the scalar case, which is
developed in~\cite{teschl-jacobioperatorscompletely-1999,akhiezer-classicalmomentproblem-2020,schmudgen-momentproblem-2017}. 

\subsection{Block Jacobi operators}\label{sec:jacobi}
We start by fixing the notation.
For $L \in \nnum$, we denote the vector space of $\cnum^L$-valued sequences $u = (u_n)_{n \in \nnum}$ by $\ell(\nnum,\cnum^L)$.
The subspace of square-summable sequences $\ell^2(\nnum,\cnum^L)$ is a Hilbert space with respect to the scalar product
\begin{equation}
	\braket{u,v} = \sum_{n = 0}^\infty \braket{u_n,v_n} \, , \quad \norm{u} = \sqrt{\braket{u,u}}\, .
\end{equation}
Moreover, we introduce the space of finitely supported sequences
\begin{equation}
	\ell_0(\nnum,\cnum^L) = \left\{u \in \ell(\nnum,\cnum^L) \, : \, \exists N \in \nnum \: \text{s.t.} \:  u_n = 0 \quad \forall n \geq N\right\},
\end{equation}
which is a dense subspace of $\ell^2(\nnum,\cnum^L)$. 
We also require the projections $\pi_n: \ell(\nnum,\cnum^L)\to\cnum^L$, $\pi_n u = u_n$, onto the $n$th entry of a sequence.
Finally, let $\Mat(L,\cnum)$ be the space of complex $L \times L$ matrices acting on $\cnum^L$, which we equip with the supremum norm $\norm{\cdot}_\infty$.
These definitions allow us to introduce the block Jacobi operators, which act on sequences $u \in \ell(\nnum,\cnum^L)$ as the infinite-dimensional block matrix
\begin{equation}
    \label{eq:def-block-jacobi}
    \Jop = \begin{pmatrix}
            B_0 & A_0 & 0 & 0 & 0 & \dots \\
            A_0^\ast & B_1 & A_1 & 0 & 0 & \dots \\
            0 & A_1^\ast & B_2 & A_2 & 0 & \dots \\
            0 & 0 & A_2^\ast & B_3 & A_3 & \dots \\
            \vdots & \vdots & \vdots & \vdots & \vdots  & \ddots
           \end{pmatrix} \, ,
\end{equation}
where $A_n,B_n \in \Mat(L,\cnum)$ and the $B_n$ are Hermitian.
Equivalently, we can define the action of $\Jop$ on a sequence $u \in \ell(\nnum,\cnum^L)$ by
\begin{align}
    (\Jop u)_n & = A_n u_{n+1} + B_n u_n + A_{n-1}^\ast u_{n-1}\,, \quad n \geq 1\,, \\
    (\Jop u)_0 & = A_0 u_1 + B_0 u_0.
\end{align}
If either of the sequences $(A_n)_{n \in \nnum}$ or $(B_n)_{n \in \nnum}$ is unbounded with respect to the supremum norm, $\Jop$ cannot define a bounded operator on $\ell^2(\nnum,\cnum^L)$.
We hence define the minimal operator $J$ as the restriction of $\Jop$ to the initial domain of finitely supported sequences $\ell_0(\nnum,\cnum^L)$, which is a Hermitian operator because the $B_n$ are Hermitian by assumption.
Our goal is to determine whether $J$ is essentially self-adjoint, and, if not, to determine its deficiency indices.
To this end, we assume in the following that the off-diagonal blocks $A_n$ are invertible.
We then obtain a first sufficient condition for essential self-adjointness of $J$, the so-called Carleman criterion:
\begin{proposition}
    \label{prop:carleman}
    Let $J$ be a block Jacobi operator defined by \cref{eq:def-block-jacobi} on $\ell_0(\nnum,\cnum^L)$, and assume that the Carleman condition
    \begin{equation}
        \sum_{n = 0}^\infty \frac{1}{\norm{A_n}_\infty} = \infty
    \end{equation}
    holds.
    Then $J$ is essentially self-adjoint.
\end{proposition}
Finding necessary conditions for essential self-adjointness, or determining the deficiency spaces $\dspace_\pm = \Ran(J\pm \iu)^\perp$, is more involved, as we now discuss.
To this end, we introduce the \textit{matrix-valued generalized eigenvalue equation}: for $z \in \cnum$,
\begin{equation}
    \label{eq:generalized-eigenvalue}
    A_n U_{n + 1} + B_n U_n + A_{n-1}^\ast U_{n-1} = z U_n \quad \forall n \geq 1\, ,
\end{equation}
for sequences $(U_n)_{n \in \nnum}$ of matrices $U_n \in \Mat(L,\cnum)$.
Since the $A_n$ are invertible, specifying the initial conditions $U_0,U_1 \in \Mat(L,\cnum)$ determines a unique solution iteratively.
In analogy with the \textit{orthogonal polynomials of the first kind} in the scalar case (namely, the solutions satisfying the Dirichlet condition $P_{-1}(z)=0$), we introduce the particular solution $(P_n(z))_{n \in \nnum}$ by fixing the initial conditions
\begin{equation}
    \label{eq:poly-p}
    P_0(z) = \id \, , \quad P_1(z) = A_0^{-1}(z-B_0) \, .
\end{equation}
It follows that $P_n(z)$ is a matrix-valued polynomial in $z$ of degree $n$, which solves the matrix-valued eigenvalue equation $\Jop P(z) = z P(z)$.
In particular, for all $x \in \cnum^L$, the sequence $(P_n(z) x)_{n \in \nnum}$ is a solution of the equation $\Jop u = z u$.
Since the deficiency spaces of $J$ are exactly the square-summable solutions of this equation, we obtain the following characterization:
\begin{lemma}
    \label{lem:dspace}
    Let $J$ be a block Jacobi operator defined as in \cref{eq:def-block-jacobi}, and $P(z)$ the associated orthogonal polynomials of the first kind.
    The deficiency spaces of $J$ are given by
    \begin{equation}
        \label{eq:dspace-poly}
        \dspace_\pm = \left\{ (P_n(\pm \iu) x)_{n \in \nnum} \, : \, x \in \cnum^L, \;\sum_{n = 0}^\infty \norm{P_n(\pm \iu) x}^2 < \infty \right\} \, .
    \end{equation}
\end{lemma}
For a block Jacobi operator, one thus obtains $0 \leq n_\pm \leq L$.
On the other hand, for all numbers $0 \leq n_\pm \leq L-1$ there exists a block Jacobi operator with deficiency indices $n_\pm$ \cite{dyukarev-deficiencynumberssymmetric-2006}; in particular, block Jacobi operators with unequal deficiency indices exist.
Determining the dimension of $\dspace_\pm$ directly from \cref{eq:dspace-poly} is, however, difficult in general.

\subsection{Why counting solutions is not enough}\label{sec:counting}
We now explain why determining the deficiency indices of a block Jacobi operator directly from \cref{lem:dspace}, whether numerically or analytically, is not straightforward.
Reformulating \cref{lem:dspace}, consider the generalized eigenvalue equation
\begin{equation}
    \label{eq:generalized-eigenvalue-vector}
    A_n u_{n+1}+B_n u_n + A_{n-1}^\ast u_{n-1} = \pm\iu u_n \, , \quad n \geq 1 \, ,
\end{equation}
with $u \in \ell(\nnum,\cnum^L)$ (cf. \Cref{eq:generalized-eigenvalue}).
We denote by $K_\pm \subset \ell(\nnum,\cnum^L)$ the space of all solutions of \cref{eq:generalized-eigenvalue-vector}, and set
\begin{equation}
    K^p_\pm = \{ u \in K_\pm \, : \, u_1 =  -A_0^{-1}(B_0 \mp \iu) u_0 \} \, .
\end{equation}
We thus obtain $\dim K^p_\pm = L$ and
\begin{equation}
    \dspace_\pm = \{u \in K^p_\pm \, : \, u \in \ell^2(\nnum,\cnum^L)\} \, .
\end{equation}
For a given initial condition $u_0 \in \cnum^L$, it is easy to check numerically whether the associated solution $u \in K^p_\pm$ is square summable.
If, however, $\Jop$ is not maximally indeterminate---if $n_+ < L$, say---the numerical analysis of $\dspace_+$ becomes intractable: the square-summable solutions then form an $n_+$-dimensional subspace of $K^p_+$, a set of measure zero, and locating it numerically is unstable.
Counting square-summable solutions in $K^p_\pm$ is therefore not a viable route to the deficiency indices.

Another approach is to determine $2L$ linearly independent asymptotic solutions $u \in K_\pm$ of \cref{eq:generalized-eigenvalue-vector} analytically, using Birkhoff's theorem~\cite{birkhoff-formaltheoryirregular-1930,birkhoff-analytictheorysingular-1933,wimp-resurrectingasymptoticslinear-1985} or WKB methods \cite{olver-asymptoticsspecialfunctions-2010,fedoriuk-asymptoticanalysislinear-1993}.
Analyzing whether these solutions are square summable, however, is not sufficient.
First, not all $u \in K_\pm$ satisfy the initial condition $u_1 =  -A_0^{-1}(B_0 \mp \iu) u_0$, so counting the linearly independent solutions $u \in K_\pm \cap \ell^2(\nnum,\cnum^L)$ does not suffice.
Second, if $u,v \in K_\pm$ are linearly independent and neither is square summable, $u+v$ may nevertheless be square summable.
Counting the square-summable asymptotic solutions is therefore not guaranteed to yield the correct result.
Finally, both the statement and the proof of Birkhoff's theorem are notoriously difficult to parse~\cite{wimp-resurrectingasymptoticslinear-1985,wong-asymptoticexpansionssecondorder-1992}, and the reduction to the well-understood three-term case~\cite{wong-asymptoticexpansionssecondorder-1992,elaydi-introductiondifferenceequations-2005} used for higher-order squeezing operators in \cite{fischer-selfadjointrealizationshigherorder-2026} is not available for general polynomials.
All of this makes it difficult to obtain a general characterization of the deficiency spaces $\dspace_\pm$ by either numerical or analytical asymptotic methods.

The Weyl disk approach developed in \cref{sec:weyl-disks}, by contrast, is numerically stable, as it does not depend on selecting the correct initial conditions.
Instead, it produces a sequence of radius matrices $R_n^\pm$ converging monotonically to a limiting matrix $R^\pm$, and determining the deficiency indices reduces to counting the nonzero eigenvalues of $R^\pm$.

\subsection{Weyl disks and deficiency indices}\label{sec:weyl-disks}
We now present an approach generalizing the scalar Weyl's limit-point/limit-circle dichotomy to the matrix case, referring to~\cite{schulz-baldes-geometryweyltheory-2010} for a full development of the theory.

To begin with, we consider the truncation of \cref{eq:def-block-jacobi} after the $n$th block, with a boundary matrix $T\in\Mat(L,\cnum)$ added at the last site:
\begin{equation}
        J_n(T) = \begin{pmatrix}
            B_0 & A_0 & & & & & \\
            A_0^\ast & B_1 & A_1 &  &  &  & \\
            & A_1^\ast & B_2 & A_2 &  &  & \\
            &  & A_2^\ast & B_3 &   \ddots & & \\
            &  &  &   \ddots  & \ddots & \ddots & \\
            &  &  &  & \ddots & B_{n-1} & A_{n-1} \\
            & &   &  & &  A_{n-1}^\ast & B_n + T 
           \end{pmatrix} \, .
\end{equation}
For Hermitian $T$, $J_n(T)$ is a finite-dimensional Hermitian matrix, and hence its resolvent $(J_n(T)-z)^{-1}$ exists for all complex energies $z \in \cnum \setminus \rnum$.
Of particular interest is the first diagonal block of the resolvent,
\begin{equation}
    G_n(z,T) = \pi_0 (J_n(T)-z)^{-1} \pi_0^\ast \in \Mat(L, \cnum) \, ,
\end{equation}
where $\pi_0$ is the projection onto the first block.
The matrix-valued function $G_n(z,T)$ is known as the Green function
of $J_n(T)$, and effectively contains its entire spectral theory, as does the Weyl--Titchmarsh function of scalar Jacobi and Sturm--Liouville operators (cf.~\cite{schmudgen-unboundedselfadjointoperators-2012}).
As the $J_n(T)$ are finite-dimensional approximations of $\Jop$, we are interested in how the boundary matrix $T$ influences the Green function $G_n(z,T)$ as $n \to \infty$.
For simplicity, we restrict ourselves to the case $\Imag z>0$; the case $\Imag z < 0$ is analogous up to signs, cf.~\cite{schulz-baldes-geometryweyltheory-2010}.
We shall also consider boundary matrices $T$ from the matrix upper half-plane
\begin{equation}
    \matupper(L) = \left\{ T \in \Mat(L,\cnum) \, : \, \iu (T^\ast-T) > 0\right\} \, .
\end{equation}
The boundary of the matrix upper half-plane is a stratified space, but, as in the scalar case, its maximal boundary, in the sense of~\cite{schulz-baldes-geometryweyltheory-2010}, consists exactly of the Hermitian matrices $T = T^\ast$, which we denote by $\partial\matupper(L)$. The Green function then admits the following characterization:
\begin{proposition}
    \label{prop:weyl-sphere}
    For $\Imag z > 0$ and $n \geq 2$, let 
    \begin{equation}
        \sphere_n^z = \left\{ G_n(z,T) \, : \, T \in \matupper(L) \right\}
    \end{equation}
    be the $n$th Weyl disk associated with $\Jop$.
    Then the following holds:
    \begin{enumerate}[(i)]
        \item $\sphere_n^z \subset \matupper(L)$, i.e., the Weyl disk is contained in the matrix upper half-plane;
        \item $T\mapsto G_n(z,T)$ is a diffeomorphism between $\matupper(L)$ and $\sphere_n^z$;
        \item the Hermitian boundary matrices $T=T^\ast$ correspond, up to a set of measure zero $N$, to the maximal boundary $\partial \sphere_n^z$ of $\sphere_n^z$, i.e.,
        \begin{equation}
            \partial  \sphere_n^z = \left\{ G_n(z,T) \, : \, T \in \Mat(L,\cnum), \, T^\ast = T \right\}\cup \,N \, .
        \end{equation}
    \end{enumerate}
\end{proposition}
We have thus obtained a geometric description of the influence of the boundary matrix $T$ on the truncation $J_n(T)$:
as $T$ ranges over all Hermitian matrices, $G_n(z,T)$ ranges over the boundary $\partial \sphere_n^z$ of the Weyl disk $\sphere_n^z$.
We now show that the term \textit{disk} is justified, starting with the radius matrices.
\begin{lemma}
    \label{lem:radius}
    Let $J$ be a block Jacobi operator defined as in \cref{eq:def-block-jacobi} and $P(z)$ the associated orthogonal polynomials of the first kind defined in \cref{eq:poly-p}.
    For $n \in \nnum$ and $\Imag z \neq 0$, let the radius matrix $R_n^z$ be given by
    \begin{equation}
        R_n^z = \left(2 |\Imag z| \sum_{k = 0}^n P_k(z)^\ast P_k(z) \right)^{-1} \, .
    \end{equation}
    Then the following holds:
    \begin{enumerate}[(i)]
        \item the $R_n^z$ are well defined, strictly positive definite Hermitian matrices;
        \item $R_{n+1}^z \leq R_n^z$ and $R_{n+2}^z < R_n^z$.
    \end{enumerate}
\end{lemma}
In the literature, the radius matrices for $\Imag z < 0$ sometimes carry an additional sign, which we omit for ease of notation.
We also write $R^\pm_n = R^{\pm \iu}_n$ and $\sphere^\pm_n = \sphere^{\pm \iu}_n$, as in the main text.
Analogously to the radius matrices $R_n^z$, one can define center matrices $S_n^z$; since we do not need their explicit form, we refer to~\cite{schulz-baldes-geometryweyltheory-2010} for their definition.
The radius and center matrices yield the following characterization of the Weyl disks, which justifies the terminology:
\begin{proposition}
    \label{prop:radius-center-sphere}
    For $\Imag z > 0$ and $n \in \nnum$, the boundary of the Weyl disk is
    \begin{equation}
        \partial \sphere_n^z = \left\{ S_n^z + (R_n^z)^{1/2} W (R_n^{\bar z})^{1/2} \, : \, W \in U(L) \right\} \, ,
    \end{equation}
    where $U(L)$ denotes the group of $L\times L$ unitary matrices.
\end{proposition}
As the radius matrices $R_n^z$ are strictly positive (cf.~\cref{lem:radius}), the $\sphere_n^z$ are thus matrix disks in $\matupper(L)$, with center $S_n^z$ and radii $R_n^z$ and $R_n^{\bar z}$, whose maximal boundaries are parametrized by the unitaries $W \in U(L)$. This parametrization also accounts for the set $N$ in \cref{prop:weyl-sphere}(iii): the Cayley transform maps the Hermitian matrices onto the unitaries without the eigenvalue $-1$, so the Hermitian boundary conditions reach all of $\partial\sphere_n^z$ but a set of measure zero.
By \cref{lem:radius}, their radii decrease monotonically, so the Weyl disks $\sphere_n^z$ shrink as $n \to \infty$.
In fact, a stronger statement holds:
\begin{lemma}
    \label{lem:nested}
    For $\Imag z > 0$ and $n \in \nnum$, let $\sphere_n^z$ be the $n$th Weyl disk as defined in \cref{prop:weyl-sphere}.
    Then $\sphere_{n+1}^z \subset \sphere_n^z$ and $\partial \sphere_{n-1}^z \cap \overline{\sphere_{n+1}^z} = \emptyset$, i.e., the Weyl disks are strictly nested.
\end{lemma}
\Cref{lem:nested} implies that the limit disk
\begin{equation}
    \sphere^z = \bigcap_{n \in \nnum} \sphere_n^z
\end{equation}
always exists; it need not, however, have full dimension, i.e., the disks may collapse onto a lower-dimensional disk, or even onto a single point of $\matupper(L)$.
In the \textit{limit-point} case, the boundary matrix $T$ ceases to influence $G_n(z,T)$ as $n \to \infty$, and all information about the spectrum and dynamics of $\Jop$ is effectively contained in finite dimensions.
If, instead, a (possibly lower-dimensional) \textit{limit circle} survives, $T$ influences $G_n(z,T)$ for arbitrarily large $n$, and an explicit choice must be made.
The following proposition makes the connection with von Neumann's theory of deficiency indices explicit, taking both half-planes into account.
\begin{proposition}
    \label{prop:radius}
    Let $\Jop$ be a block Jacobi operator defined as in \cref{eq:def-block-jacobi}, and, for $n \in \nnum$, let $R_n^\pm = R_n^{\pm \iu}$ be the associated radius matrices from \cref{lem:radius}.
    Then the limits $R^\pm = \lim_{n \to \infty} R_n^\pm$ exist.
    Moreover, the following holds:
    \begin{enumerate}[(i)]
        \item $R^\pm$ are positive semidefinite;
        \item the deficiency indices of $J$ are given by $n_\pm = \dim \Ran R^\pm$.
    \end{enumerate}
\end{proposition}
    Similarly, one can define the limit center matrices $S^z$, the limiting Green functions $G(z,T)$, and the associated Nevanlinna functions.
    For a more thorough development, we refer the reader to~\cite{schulz-baldes-geometryweyltheory-2010,dyukarev-deficiencynumberssymmetric-2006,moszynski-blockjacobimatrices-2025}.
    
We have thus obtained a geometric interpretation of the deficiency indices.
If $n_+ = n_- = 0$ ($J$ is essentially self-adjoint), the Weyl disks, which describe the possible Green functions $G_n(z,T)$ for a boundary matrix $T$ at site $n$, collapse to a single point as $n \to \infty$, and no boundary information is needed in the infinite-dimensional limit.
If $n_+ = n_- \neq 0$ ($J$ has multiple self-adjoint extensions), the Weyl disks do not collapse to a single point, and different boundary matrices yield different Green functions, and thus different spectra and different dynamics.
In this case, $n_\pm$ is the dimension of the limit disk $\sphere^\pm$; in particular, the deficiency spaces $\dspace_\pm$ are given explicitly in terms of $R^\pm$ as
\begin{equation}
    \dspace_\pm  = \left\{ P(\pm \iu) x \, : \, x \in (\Ker R^\pm)^\perp \right\} \, .
\end{equation}
Instead of by boundary matrices $T$, the self-adjoint extensions of $J$ can then be parametrized explicitly, following von Neumann, by the unitaries $V: \dspace_+ \to \dspace_-$.
Finally, if $n_+ \neq n_-$ ($J$ has no self-adjoint extensions), the limit disks $\sphere^+$ and $\sphere^-$ have different dimensions, and no single Hermitian boundary condition can be chosen for both half-planes.

\subsection{Probability flux at infinity}\label{sec:flux}
Essential self-adjointness, and its failure, can also be interpreted in terms of the probability flux to infinity in finite time.
Let $J$ be a block Jacobi operator with equal deficiency indices, let $J_V$ be any of its self-adjoint extensions, and let $\psi(t) = \e^{-\iu J_V t} \psi(0)$ be the time evolution under $J_V$ of an initial state $\psi(0) \in \domain(J_V)$.
For $n \in \nnum$, let moreover $\Pi_n = \sum_{k = 0}^n \pi_k^\ast\pi_k$ be the projector onto the first $n+1$ blocks.
The rate of change of the probability contained in these blocks is
\begin{equation}\label{eq:flux}
    F_n(t) = \diff*{}{t}\braket{\psi(t),\Pi_n \psi(t)} \, ,
\end{equation}
and we are interested in its limit $F(t) = \lim_{n \to \infty} F_n(t)$, which quantifies the probability flux to infinity in finite time.
An explicit calculation shows that 
\begin{equation}
    F_n(t) = \iu\braket{\psi(t),[J_V, \Pi_n] \psi(t)} = \iu \left(\braket{A_n \psi_{n+1}(t), \psi_n(t)}-\braket{\psi_{n}(t),A_n\psi_{n+1}(t)}\right) = \wronsk(\Psi_n(t),\Psi_n(t))\, ,
\end{equation}
where the Wronskian $\wronsk(x,y)$ of $x,y \in \cnum^{2L}$ and the vectors $\Psi_n(t)$ are given by
\begin{equation}
    \wronsk(x,y) = \frac{1}{\iu} x^\ast \begin{pmatrix}
        0 & -\id \\ \id & 0 
    \end{pmatrix}y \, , \quad \Psi_n(t) = \begin{pmatrix}A_n \psi_{n+1}(t)\\\psi_n(t)\end{pmatrix}\, .
\end{equation}
The Wronskian $\wronsk$ is closely related to a geometric description of the Weyl disks $\sphere_n^z$ in terms of isotropic subspaces of quadratic forms; see~\cite{schulz-baldes-geometryweyltheory-2010} for details.
In particular, $\lim_{n \to \infty} \wronsk(\Psi_n,\Psi_n) = 0$ for all $\psi \in \domain(J_V)$
if and only if $J$ is essentially self-adjoint, i.e., if and only if $J_V$ is the unique self-adjoint extension of $J$.
This yields a probabilistic interpretation of essential self-adjointness.
If $J$ is essentially self-adjoint, the probability flux to infinity in finite time, $F(t)$, always vanishes, and no additional information needs to be supplied.
If, instead, $J$ is not essentially self-adjoint but has equal deficiency indices, $F(t)$ is nonzero for some initial states $\psi(0)$, and an explicit boundary condition $V$ is needed to fix how a wave packet is reflected at infinity.
A similar interpretation in terms of the corresponding classical Hamiltonian flow is given in~\cite{zhu-classicalsymptomsquantum-1993}.

\section{From bosonic polynomials to block Jacobi operators}\label{sec:pm-domains}

This appendix connects the preceding one to the operators of the main text.
In \cref{sec:domains} we settle what $p_M$ \emph{is}, showing that all natural
choices of domain lead to the same deficiency indices; in \cref{sec:reduction}
we exhibit a unitary equivalence between $p_M$ and a block Jacobi operator; and
in \cref{sec:main-result} we combine the two with \cref{prop:radius} to obtain
\cref{cor:main}, the statement the algorithm implements.

\subsection{The choice of domain}\label{sec:domains}
The main text treats $p_M$ as an operator on $L^2(\rnum)$.
Strictly speaking, \cref{eq:def-pm} is a formal
polynomial in $a$ and $a^\dagger$,
\begin{equation}
   \label{eq:def-pm}
   p_M = \sum_{k,l=0}^s M_{k,l} (a^\dagger)^k a^l \, ,
\end{equation}
with $M \in \Mat(s+1,\cnum)$, and it becomes an operator only once a dense
domain of vectors on which it acts has been fixed.
Which vectors these should be is not dictated by the physics.
Writing $p_M$ in terms of $a$ and $a^\dagger$ suggests working with the Fock
states, on which the action of \cref{eq:action-beta} is elementary; writing the
same operator as a polynomial in $x$ and $p$ instead suggests the domains familiar
from the theory of differential operators.
At least four natural candidates therefore present themselves, and a priori
they define four different operators.

The following lemma shows that they do not: all four have the same closure, and
hence the same deficiency indices and the same self-adjoint extensions.
The answer to the question posed in the main text---how many unitary evolutions does $p_M$
generate?---is thus the same for any of these choices, and
depends on $M$ alone.

\begin{lemma}
    \label{lem:all-domains}
    Let $M \in \Mat(s+1,\cnum)$ be Hermitian, and let $p_M$ be the formal
    polynomial in \cref{eq:def-pm}.
    Consider $p_M$ as an operator on $L^2(\rnum)$ with any one of the following
    domains:
    \begin{enumerate}[(i)]
        \item the finite linear span of the Fock states $(\phi_n)_{n \in \nnum}$;
        \item the Schwartz space $\schwartz(\rnum)$;
        \item the smooth compactly supported functions $C_0^\infty(\rnum)$;
        \item for any $k \geq 2s$, the vectors
        $\sum_{n = 0}^\infty c_n \phi_n$ with
        $\sum_{n = 0}^\infty n^k |c_n|^2 < \infty$.
    \end{enumerate}
    Each of these operators is densely defined, and Hermitian if and only if
    $M$ is.
    Moreover, they all have the same closure, and therefore the same deficiency
    indices and the same self-adjoint extensions.
\end{lemma}

\begin{proof}
    Density and Hermiticity are immediate in each case, the latter from
    \Cref{eq:action-beta}.
    For the remaining claim, write $\domain_F$, $\domain_S$, $\domain_C$ and
    $\domain_k$ for the four domains above, in order, and let
    $\overline{X}^{\norm{}_1}$ denote the closure of a set $X$ with respect to
    the graph norm
    \begin{equation}
        \norm{\psi}_1^2 = \norm{p_M \psi}^2 + \norm{\psi}^2 \, .
    \end{equation}
    We must show that
    \begin{equation}
        \label{proofeq:domains-equal}
        \overline{\domain_F}^{\norm{}_1} = \overline{\domain_S}^{\norm{}_1}
        = \overline{\domain_C}^{\norm{}_1} = \overline{\domain_k}^{\norm{}_1}
    \end{equation}
    for all $k \geq 2s$.
    Since $\domain_F \subset \domain_S \subset \domain_k$ and
    $\domain_C \subset \domain_S$, and the closure of a set is the smallest
    closed set containing it, it suffices to prove the two inclusions
    \begin{equation}
        \label{proofeq:closure-incusion}
        \overline{\domain_F}^{\norm{}_1} \supset \domain_k \, , \quad
        \overline{\domain_C}^{\norm{}_1} \supset \domain_S \, .
    \end{equation}
    For the first one, let $\psi = \sum_{n = 0}^\infty c_n \phi_n \in \domain_k$
    and set $\psi_m = \sum_{n = 0}^m c_n \phi_n \in \domain_F$, so that
    \begin{equation}
        \label{proofeq:tail}
        \lim_{m \to \infty} \norm*{(a^\dagger a)^{k/2} (\psi-\psi_m)}^2
        = \lim_{m \to \infty} \sum_{n = m+1}^\infty n^k |c_n|^2 = 0 \, ,
    \end{equation}
    by the definition of $\domain_k$.
    As $p_M$ is a polynomial of degree $\deg p_M \leq 2s$, there exist
    $C_1,C_2>0$ such that
    \begin{equation}
        \norm{p_M \varphi}^2 \leq
        C_1 \norm*{(a^\dagger a)^{k/2} \varphi}^2+C_2 \norm{\varphi}^2
    \end{equation}
    for all $\varphi \in \domain_F$, and therefore
    \begin{equation}
        \norm{\psi-\psi_m}_1^2 \leq
        C_1 \norm*{(a^\dagger a)^{k/2} (\psi-\psi_m)}^2
        + (C_2+1) \norm{\psi-\psi_m}^2 \, .
    \end{equation}
    The difference $\psi - \psi_m$ is supported on the indices
    $n \geq m+1 \geq 1$, on which $n^k \geq 1$; hence
    \begin{equation}
        \norm{\psi-\psi_m}^2 = \sum_{n = m+1}^\infty|c_n|^2
        \leq \sum_{n = m+1}^\infty n^k |c_n|^2
        = \norm*{(a^\dagger a)^{k/2}(\psi-\psi_m)}^2 \, ,
    \end{equation}
    and both terms vanish as $m \to \infty$ by \cref{proofeq:tail}.
    Thus $\norm{\psi-\psi_m}_1 \to 0$, and the first inclusion in
    \Cref{proofeq:closure-incusion} holds.

    For the second one, $p_M$ maps $\schwartz(\rnum)$ continuously into
    $L^2(\rnum)$ with respect to the Schwartz topology, so the graph norm
    $\norm{\cdot}_1$ is continuous on $\schwartz(\rnum)$.
    Since $C_0^\infty(\rnum)$ is dense in $\schwartz(\rnum)$ in that topology,
    it is dense in it with respect to $\norm{\cdot}_1$ as well, which is the
    second inclusion.
\end{proof}
Hereafter, and without loss of generality, we therefore write $p_M$ for the operator
acting as in \cref{eq:def-pm} on the finite linear span of the Fock states
$(\phi_n)_{n \in \nnum}$---the choice in which the block Jacobi structure is
manifest, as we now show.

\subsection{Reduction to block Jacobi form}\label{sec:reduction}
We first recall some definitions
from \cite[Section 2]{fischer-selfadjointrealizationshigherorder-2026}.
For $x \in \mathbb{Z}$ and $y \in \nnum$, the Pochhammer symbol $(x,y)$ is given by
\begin{equation}
    (x,y) = \prod_{i = 1}^y (x+i-1) = x (x+1) \dots (x+y-1) \,,
\end{equation}
with the convention $(x,y) = 0$ for $x<0$.
Moreover, for $n,k,l \in \nnum$, we introduce the coefficients
\begin{equation}
    \label{eq:beta}
    \beta^{k,l}_n = \sqrt{(n-l+1,l)(n-l+1,k)} \, ,
\end{equation}
which obey the symmetry relation $\beta^{k,l}_{n-k+l} = \beta^{l,k}_n$.
In particular, we obtain
\begin{equation}
    \label{eq:action-beta}
    (a^\dagger)^k a^l \phi_n = \beta^{k,l}_n \phi_{n + k-l}
\end{equation}
for all $n \in \nnum$.
Here and in the following, we abuse notation by allowing Fock states $\phi_n$ with negative $n$ on the right-hand side, noting that $\beta^{k,l}_n = 0$ for all $n < l$.
Moreover, we assume that $p_M$ is not diagonal, i.e., not a function of the number operator $a^\dagger a$, as such operators are trivially essentially self-adjoint~\cite[Theorem 2.21]{teschl-mathematicalmethodsquantum-2009}.
These definitions allow us to make the following statement, which is similar to the one obtained in \cite{mirzoev-generalizedjacobimatrices-2021}.
\begin{proposition}
    \label{prop:decomposition}
    Let $M \in \Mat(s+1,\cnum)$ be a Hermitian matrix, let $p_M$ be the operator from \cref{eq:def-pm}, and set $d = \max\{|k-l|\, :\, M_{k,l} \neq 0\}$.
    Then $p_M$ is unitarily equivalent to the block Jacobi operator $J_M$ on $\ell^2(\nnum,\cnum^d)$, defined as in \cref{eq:def-block-jacobi} by the matrices
    \begin{alignat}{2}
        A_n &= \left((A_n)_{ij}\right)_{i,j=0,\dots,d-1}\, , & \quad (A_n)_{ij}  & =  \sum_{l=d-i+j}^{s} M_{l-d+i-j,l} \beta^{l,l-d+i-j}_{nd+i} \, , \\
        B_n &= \left((B_n)_{ij}\right)_{i,j=0,\dots,d-1}\, , & \quad (B_n)_{ij}  & =  \sum_{l=\max(j-i,0)}^{\min(s,s+j-i)} M_{l+i-j,l} \beta^{l+i-j,l}_{nd + j} \, .
    \end{alignat}
    In particular, the $B_n$ are Hermitian for all $n \in \nnum$, and the $A_n$ are lower triangular matrices.
\end{proposition}
\begin{proof}
    By assumption, $p_M$ is not diagonal and thus $d \geq 1$ holds. 
    For $n \in \nnum$ and $i=0,\dots,d-1$, let $e_n^i$ be the sequence defined by
    \begin{equation}
        (e_n^i)_m = \delta_{nm} e^i\,,
    \end{equation}
    where $e^i$ is the $i$th basis vector of $\cnum^d$.
    Then the unitary operator $U: \ell^2(\nnum,\cnum^d) \to L^2(\rnum)$ defined by
    \begin{equation}
        \label{proofeq:def-u}
        U e_n^i = \phi_{nd+i}
    \end{equation}
    and extended by linearity, implements the unitary equivalence between $\ell^2(\nnum,\cnum^d)$ and $L^2(\rnum)$.
    In particular, it maps $\ell_0(\nnum,\cnum^d)$ onto $\mspan(\phi_r)_{r \in \nnum}$.
    In the following, we show that $J_M = U^\ast p_M U$.
       Using \cref{eq:action-beta}, we immediately get
    \begin{equation}
        p_M \phi_r = \sum_{k,l=0}^s \beta^{k,l}_r M_{k,l}\phi_{r+k-l}\, .
    \end{equation}
    By definition, $\max \{|k-l| \, : \, M_{k,l} \neq 0 \} = d$, so only the terms with $|k-l| \leq d$ contribute.
    Grouping by constant difference $\Delta = k-l$, we obtain
    \begin{equation}
        \label{proofeq:pm-expansion}
         p_M \phi_r = \sum_{\Delta = -d}^d \sum_{l=\max(-\Delta,0)}^{\min(s,s-\Delta)} M_{l+\Delta,l} \beta^{l+\Delta,l}_r \phi_{r + \Delta} \,.
    \end{equation}
    Our goal is to identify $p_M$ with the block Jacobi operator $J_M$ by means of the unitary transformation $U$ of \cref{proofeq:def-u}.
    To this end, we compute the scalar product
   $\braket{\phi_{m d + i},p_M \phi_{n d + j}}$
    for $m,n \in \nnum$ and $i,j \in \{0,\dots,d-1\}$.
    Substituting $r = nd + j$ into \cref{proofeq:pm-expansion}, we get
    \begin{equation}
         p_M \phi_{nd+j} = \sum_{\Delta = -d}^d \sum_{l=\max(-\Delta,0)}^{\min(s,s-\Delta)} M_{l+\Delta,l} \beta^{l+\Delta,l}_{nd+j} \phi_{nd+j + \Delta} \,.
    \end{equation}
    It is hence clear that $ \braket{\phi_{md + i},p_M \phi_{nd + j}} \neq 0$ implies $m\in \{n,n\pm 1\}$.
    In the following, we treat the cases $m=n$ and $m = n\pm 1$ separately.
    For $m=n$, we get non-zero contributions for $\Delta = i-j$, and hence
    \begin{equation}
        \label{proofeq:bmat}
        \braket{\phi_{n d + i},p_M \phi_{nd + j}} = \sum_{l=\max(j-i,0)}^{\min(s,s+j-i)} M_{l+i-j,l} \beta^{l+i-j,l}_{nd + j} \, .
    \end{equation}
    Similarly, for $m=n-1$ the summands with $\Delta = -d+i-j$ contribute and, using $\Delta \leq -1$, we get
    \begin{equation}
        \braket{\phi_{(n-1)d + i}, p_M \phi_{nd + j}} = \sum_{l=d-i+j}^{s} M_{l-d+i-j,l} \beta^{l-d+i-j,l}_{nd+j} \, ,
    \end{equation}
    and, using $\beta^{l+\Delta,l}_{nd+i - \Delta}=\beta^{l,l+\Delta}_{nd+i}$ (cf. \cite[Lemma 2.6]{fischer-selfadjointrealizationshigherorder-2026}),
    \begin{equation}
        \label{proofeq:amat}
        \braket{\phi_{nd + i}, p_M \phi_{(n+1)d + j}} = \sum_{l=d-i+j}^{s} M_{l-d+i-j,l} \beta^{l,l-d+i-j}_{nd+i}
    \end{equation}
    follows.
    In order to prove the equality $J_M = U^\ast p_M U$, we evaluate, using \cref{eq:def-block-jacobi,proofeq:def-u},
    \begin{equation}
        \braket{\phi_{md+i},U J_M U^\ast \phi_{nd+j}}
        = \begin{cases}
            (A_{n-1})_{ij} &  m=n-1 \\
            (B_{n})_{ij} & m=n \\
            (A_{n}^\ast)_{ij} & m = n+1 \\
            0 & \text{else}
        \end{cases}\, .
    \end{equation}
    Comparing this with \cref{proofeq:bmat,proofeq:amat} yields the claim.
    By construction, $B_n^\ast = B_n$ and $\braket{\phi_{(n+1)d+i},p_M\phi_{nd+j}} = (A^\ast_n)_{ij}$ for all $n \in \nnum$.

    Finally, for $j > i$ we get $|l-d+i-j-l| > d$ for all $l \in \nnum$, and hence, by the definition of $d$, $M_{l-d+i-j,l} = 0$ for all $d-i+j \leq l \leq s$.
    It follows that $(A_n)_{ij} = 0$, that is, that $A_n$ is lower triangular.
\end{proof}
The careful reader will have noticed that nothing guarantees that the blocks $A_n$ are invertible.
The following lemma shows that this is, in fact, not a problem.
\begin{lemma}
    \label{lem:invertible}
    Let $M \in \Mat(s+1,\cnum)$ be a Hermitian matrix, $p_M$ be the operator from \cref{eq:def-pm}, and let $A_n$ be given as in \cref{prop:decomposition}.
    Then there exists $N \in \nnum$ such that $A_n$ is invertible for all $n \geq N$.
\end{lemma}
\begin{proof}
    As in \cref{prop:decomposition}, let $d = \max\{|k-l|\, :\, M_{k,l} \neq 0\}$ be the maximal shift in photon number, which is precisely the size of the blocks $A_n$.
    By \cref{prop:decomposition}, the $A_n$ are lower triangular matrices, and therefore their determinant is the product of their diagonal entries:
    \begin{equation}
        \label{proofeq:det}
        \det A_n = \prod_{i = 0}^{d-1} \sum_{l = d}^s M_{l-d,l} \beta^{l,l-d}_{nd + i} \, .
    \end{equation}
    The determinant vanishes if and only if one of its factors vanishes, that is, if there exists $i=0,\dots,d-1$ such that
    \begin{equation}
        \sum_{l = d}^s M_{l-d,l} \beta^{l,l-d}_{nd + i} = 0 \, .
    \end{equation}
    In the following, let $l_{\mathrm{max}}$ be the largest $l\in\{d,\dots,s\}$ such that $M_{l-d,l}\neq0$, which exists by definition of $d$.
    By \cite[Lemma 2.10]{fischer-selfadjointrealizationshigherorder-2026}, for all $k,l \in \nnum$ we get
    \begin{equation}
        \frac{1}{2} r^{(k+l)/2} \leq \beta^{k,l}_r \leq 2 r^{(k+l)/2} \, ,
    \end{equation}
    for $r$ large enough, and it follows that there exists $N \in \nnum$ such that
    \begin{align}
        \left|\sum_{l = d}^s M_{l-d,l} \beta^{l,l-d}_{nd + i}\right| & \geq \frac{1}{2} |M_{l_{\mathrm{max}}-d,l_{\mathrm{max}}}| (nd+i)^{(2l_{\mathrm{max}}-d)/2}-\sum_{l=d}^{l_{\mathrm{max}}-1} 2 |M_{l-d,l}| (nd+i)^{(2l-d)/2} \\
        & \geq \frac{1}{4} |M_{l_{\mathrm{max}}-d,l_{\mathrm{max}}}| (nd+i)^{(2l_{\mathrm{max}}-d)/2}
    \end{align}
    holds for all $n \geq N$.
    In particular, all factors in \cref{proofeq:det} are non-zero, whence $\det A_n \neq 0$ for all $n \geq N$.
\end{proof}
We thus obtain the following decomposition of $p_M$ into a block Jacobi operator with invertible off-diagonal terms and a finite-rank perturbation:
\begin{lemma}
    \label{lem:full-decomposition}
    Let $M \in \Mat(s+1,\cnum)$ be a Hermitian matrix, $p_M$ be the operator from \cref{eq:def-pm}, and let $d$, $A_n$ and $B_n$ be defined as in \cref{prop:decomposition}.
    Then, for all $N \in \nnum$, there exists a Hermitian finite-rank operator $F$ such that $p_M$ is unitarily equivalent to the operator
    \begin{equation}
        0_N \oplus J_M(N) + F \, ,
    \end{equation}
    on $\ell^2(\nnum,\cnum^d)$, where $J_M(N)$ is the block Jacobi operator given by the sequences $(A_{n+N})_{n \in \nnum}$ and $(B_{n+N})_{n \in \nnum}$ as in \cref{eq:def-block-jacobi}, and $0_N = \bigoplus_{n = 0}^{N-1} 0_{\cnum^d}$.
    In particular, $p_M$ has the same deficiency indices as $J_M(N)$.
\end{lemma}
\begin{proof}
    For $N \in \nnum$, $0_N \oplus J_M(N)$ acts on $(u_n)_{n \in \nnum} \in \ell(\nnum,\cnum^d)$ as
    \begin{equation}
        ((0_N \oplus J_M(N)) u)_n = \begin{cases}
            0 & n \leq N-1 \\
            A_N u_{N+1} + B_N u_N & n = N \\
            A_n u_{n+1} + B_n u_n + A_{n-1}^\ast u_{n-1} & n \geq N+1
        \end{cases} \, .
    \end{equation}
    Defining $ F=J_M- 0_N \oplus J_M(N)$, we have
    \begin{equation}
        (Fu)_n = \begin{cases}
            A_0 u_1 + B_0 u_0 & n = 0 \\
            A_n u_{n+1} + B_n u_n + A_{n-1}^\ast u_{n-1} & 1 \leq n \leq N-1 \\
            A_{N-1}^\ast u_{N-1} & n = N \\
            0 & n \geq N+1
        \end{cases} \, .
    \end{equation}
    In particular, $F$ is a finite-rank operator.
    Since finite-rank perturbations do not change the deficiency indices, the claim follows.
\end{proof}

\subsection{Deficiency indices \texorpdfstring{of $p_M$}{}}\label{sec:main-result}
By \cref{lem:full-decomposition}, up to a finite-rank perturbation, every non-diagonal polynomial $p_M$ can be written as a block Jacobi operator with invertible off-diagonal blocks, and its deficiency indices can be obtained from those of $J_M(N)$ via \cref{prop:radius}.
We gather this in the following corollary, the mathematical backbone of the algorithm.

\begin{corollary}
    \label{cor:main}
    Let $M \in \Mat(s+1,\cnum)$ be Hermitian, let $p_M$ be the operator defined
    in \cref{eq:def-pm},
    and set $d = \max\{|k-l| \, : \, M_{k,l} \neq 0\}$.
    Let $A_n, B_n \in \Mat(d,\cnum)$ be the blocks of \cref{prop:decomposition},
    and let $N \in \nnum$ be such that $A_n$ is invertible for all $n \geq N$,
    as provided by \cref{lem:invertible}.
    Denote by $(P_k(z))_{k \in \nnum}$ the matrix polynomials generated by
    \begin{equation}
        P_0(z) = \id \, , \quad P_1(z) = A_N^{-1}(z-B_N) \, ,
    \end{equation}
    and by the recurrence equation \eqref{eq:generalized-eigenvalue} with the shifted
    blocks $(A_{n+N})_{n \in \nnum}$, $(B_{n+N})_{n \in \nnum}$, and set
    \begin{equation}
        \label{eq:radius-explicit}
        R^\pm = \lim_{n \to \infty}
        \left(2 \sum_{k = 0}^n P_k(\pm \iu)^\ast P_k(\pm \iu)\right)^{-1} .
    \end{equation}
    Then the limits in \cref{eq:radius-explicit} exist, and the deficiency
    indices of $p_M$ are $ n_\pm = \dim \Ran R^\pm$.
    In particular, $0 \leq n_\pm \leq d$, $n_\pm$ depends on $M$ alone, and $n_+ = n_-$ whenever $M$ is real.
\end{corollary}
\begin{proof}
    By \cref{lem:all-domains}, the deficiency indices do not depend on which of
    the initial domains is chosen.
    By \cref{prop:decomposition}, $p_M$ is unitarily equivalent to the block
    Jacobi operator $J_M$ on $\ell^2(\nnum,\cnum^d)$, and by
    \Cref{lem:invertible,lem:full-decomposition} it is unitarily equivalent to
    $0_N \oplus J_M(N)$ up to a symmetric finite-rank perturbation, whence
    $n_\pm(p_M) = n_\pm(J_M(N))$.
    Since $J_M(N)$ has invertible off-diagonal blocks, \cref{prop:radius}
    applies and gives $n_\pm(J_M(N)) = \dim \Ran R^\pm$.

    Finally, if $M$ is real, $p_M$ commutes with complex conjugation in the number basis, and $n_+ = n_-$ follows by~\cite[Theorem X.3]{reed-mmmp2-fourier-1975}.
\end{proof}
This is the statement the algorithm implements: every quantity entering
\Cref{eq:radius-explicit} is obtained from $M$ by finitely many algebraic
operations, and the only non-algebraic ingredient is the limit $n \to \infty$.

\section{The algorithm and its implementation}\label{sec:algorithm}
In this appendix, we present the algorithm in detail.
\Cref{sec:algorithm_structure} states the procedure, which implements \cref{cor:main}, step by step;
\Cref{sec:convergence,sec:rescaling,sec:raabe,sec:depth,sec:carleman-ops} then justify its individual steps, grouped by the kind of numerical difficulty each addresses rather than by order of execution;
\Cref{sec:implementation} comments on the Mathematica implementation~\cite{notebook}.

\subsection{Statement of the algorithm}\label{sec:algorithm_structure}

We collect here the steps of the procedure, in the order in which they are
executed; their individual ingredients are justified in the remainder of this
appendix.
The input is the polynomial $p_M$ of \cref{eq:def-pm}; polynomials in $x$ and
$p$ are brought into this form beforehand.
The output is the pair of deficiency indices $(n_+,n_-)$, or the statement that
no result could be reached.
Besides $M$, the procedure depends on the threshold parameters
$0<\epsilon_0<\epsilon_1$ and $r>1$, on the window length $N_c \in \nnum$, and on a
maximum number of iterations.

\begin{enumerate}[1.]
    \item \emph{Reduction to a block Jacobi operator.}
    \begin{enumerate}[(a)]
        \item Compute the largest excitation shift $d = \max\{|k-l| \, : \, M_{k,l} \neq 0\}$, which fixes the block size. If $d=0$, $p_M$ is diagonal and hence essentially self-adjoint; return $(0,0)$ and stop.
        \item Assemble the blocks $A_n, B_n \in \Mat(d,\cnum)$ from $M$ as in \cref{prop:decomposition}, so that $p_M$ is unitarily equivalent to the block Jacobi operator $J_M$ on $\ell^2(\nnum,\cnum^d)$.
        \item Test the $A_n$ for invertibility. If some fail to be invertible, determine $N$ beyond which all of them are (\cref{lem:invertible}) and pass to the shifted operator $J_M(N)$; by \cref{lem:full-decomposition} the discarded part is a finite-rank perturbation and leaves the deficiency indices unchanged.
    \end{enumerate}

    \item\label{step:shortcut} \emph{Shortcut check.}
    Test whether the operator is Gaussian, i.e. a polynomial of order $2$ or lower.
    If it is, the Carleman condition \cref{prop:carleman} holds, and the operator is essentially self-adjoint; return $(0,0)$ and stop.
    Besides saving the iteration entirely, this step sets aside the operators for which the criterion of \cref{step:decision} is known to fail; see \cref{sec:carleman-ops}.

    \item\label{step:rescaling} \emph{Rescaling.}
    Replace $p_M$ by $\lambda p_M$, with $\lambda$ fixed by normalising the coefficient of the largest excitation shift to unity.
    Since the deficiency indices are constant on each open half-plane, this leaves $n_\pm$ unchanged; it prevents the spectral parameter $z = \pm\iu$ from dominating the recurrence at small $n$ and depressing the radii artificially, as discussed in \cref{sec:rescaling}.

    \item\label{step:depth} \emph{Choice of depth.}
    Choose the sampling size $N$ large enough that the blocks $A_n,B_n$ with $n \leq N$ have reached their asymptotic regime, so that subdominant low-order contributions are not mistaken for the true asymptotics; see \cref{sec:depth}.

    \item\label{step:iteration} \emph{Iteration.}
    If $M$ is purely real, run for $z = + \iu$, if not, run for $z = + \iu$ and $z = - \iu$ separately:
    \begin{enumerate}[(a)]
        \item generate the matrix polynomials $P_k(z)$ from the initial conditions \cref{eq:poly-p} and the recurrence \cref{eq:generalized-eigenvalue};
        \item accumulate the radius matrices $R_n^\pm$ of \cref{lem:radius} up to $n = N$ and extract their eigenvalues, which by the same lemma are positive and decreasing, and hence organise into $d$ monotone bands $\rho^1_n,\rho^2_n,\dots,\rho^d_n$;
        \item monitor the numerical health of the iteration---ill-conditioning of the accumulated matrices, non-negligible imaginary parts of the $\rho^j_n$, loss of monotonicity---and, whenever one of these occurs, repeat the iteration at higher working precision, as described in \cref{sec:implementation}.
    \end{enumerate}

    \item\label{step:decision} \emph{Decision, band by band.}
    For each band, evaluate the Raabe sequence $L_n^j = n\left(1-\rho^j_{n+1}/\rho^j_n\right)$ and apply the following criteria, whose derivation and failure modes are discussed in \cref{sec:convergence,sec:raabe}:
    \begin{enumerate}[(a)]
        \item if $L_n^j > r$ for all large enough $n$, the band converges to $0$ and does not contribute;
        \item if $L_n^j < r^{-1}$ and $\rho^j_n > \epsilon_1$ for all large enough $n$, the band converges to a non-zero limit and contributes $1$;
        \item otherwise ($L_n^j<r^{-1}$ with $\rho^j_n<\epsilon_1$, or $L_n^j$ oscillating around $1$) the band is undecided: increase $N$ and return to \cref{step:iteration}, and report no result once the maximum number of iterations is reached.
    \end{enumerate}

    \item \emph{Output.}
    The deficiency indices $n_+$ and $n_-$ are the numbers of surviving bands for $z = +\iu$ and $z = -\iu$ respectively, if $M$ is purely real, $n_+=n_-$ holds.
    The result follows by \cref{prop:radius}: a unique evolution if $n_+=n_-=0$, a family of inequivalent ones if $n_+=n_-\neq0$, and none at all if $n_+\neq n_-$.
\end{enumerate}

\subsection{Deciding whether a band vanishes}\label{sec:convergence}
\begin{figure}[t]
    \centering
    \includegraphics[width=0.9\linewidth]{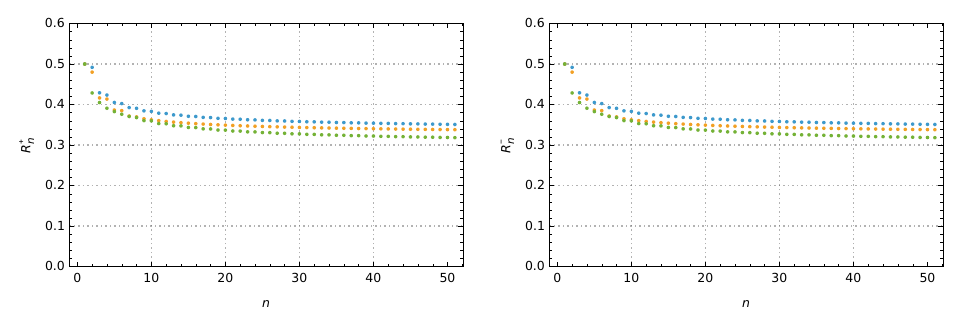}
    \caption{Eigenvalues of the radius matrices $R_n^\pm$ associated with the operator $(a^\dagger)^3 + a^3$, against $n$.}
    \label{fig:radia-sq3}
\end{figure}
This subsection and \cref{sec:raabe} justify \cref{step:decision} of \cref{sec:algorithm_structure}.
The algorithm determines the deficiency indices of $p_M$ by evaluating the eigenvalues of the radius matrices $R_n^\pm$ defined in \cref{lem:radius}.
For simplicity, we focus on a single band $\rho^j_n$, $j\in 1\dots d$,  of these eigenvalues.
By \cref{lem:radius}, the $R_n^\pm$ are strictly positive and monotonically decreasing, and the same therefore holds for the sequence $\rho^j_n$.
It follows that the limit $\rho^j = \lim_{n \to \infty} \rho^j_n$ always exists and is either $0$ or strictly positive.
By \cref{prop:radius}, the deficiency indices are given by the number of non-zero eigenvalues of the limiting radius matrices $R^\pm$.
The algorithm thus has to decide numerically whether $\rho^j>0$ for a given band, and to count the bands for which it does.
The most direct approach is to inspect the convergence visually.
As an example, we consider the trisqueezing operator $A_3 = (a^\dagger)^3 + a^3$, which is known to have deficiency indices $n_\pm = 3$ \cite{fischer-selfadjointrealizationshigherorder-2026}.
The eigenvalues of the associated radius operators $R_n^\pm$ are shown in \cref{fig:radia-sq3}.
For both $R_n^+$ and $R_n^-$, the three bands visibly converge to non-zero values, and the method reproduces $n_\pm = 3$.

We now make this criterion quantitative, so that the check can be automated rather than left to the eye.
The simplest quantitative criterion is to introduce two threshold parameters $0<\epsilon_0<\epsilon_1$ and a threshold window $N_c \in \nnum$.
We start by calculating an eigenvalue band $\rho^j_n$ for all $n \leq N$ with some $N \geq N_c$.
If the last $N_c$ values $\rho^j_n$ all lie within $\epsilon_0$ of one another, we assume that the sequence has converged sufficiently.
If they are all smaller than $\epsilon_0$, we conclude that the limit satisfies $\rho^j = 0$; if they are all larger than $\epsilon_1$, we conclude that $\rho^j> 0$.
In the intermediate case, we increase $N$.
The shortcoming of this approach is that no choice of $\epsilon_0$, however small, reliably separates $0<\rho^j<\epsilon_0$ from $\rho^j = 0$, which may lead to false claims of essential self-adjointness.
The four subsections that follow address this: two identify regimes in which the naive test fails and show how to avoid them, one introduces a sharper criterion, and one isolates a class of operators for which no numerical test is needed.

\subsection{Rescaling the operator}\label{sec:rescaling}
\begin{figure}[t]
    \centering
    \includegraphics[width=0.9\linewidth]{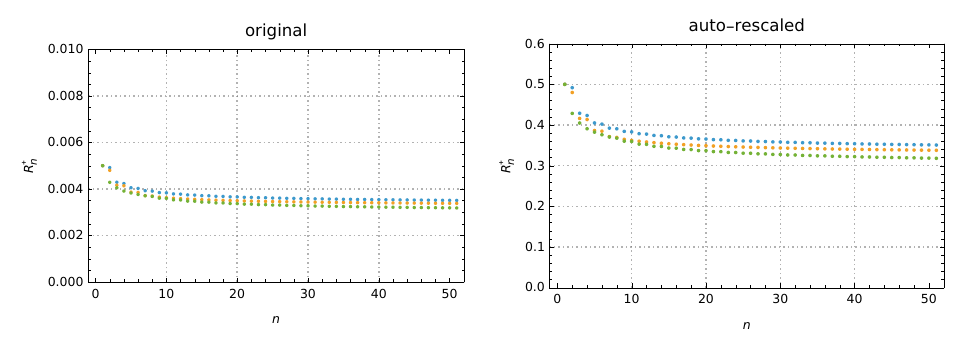}
    \caption{Eigenvalues of the radius operators $R_n^+$ associated with the operator $A_{3,\eta} =  \eta((a^\dagger)^3 + a^3)$ for $\eta = 10^{-2}$ (left) and for the automatically rescaled version ($\eta = 1$, right).}
    \label{fig:radia-sq3-prefactor}
\end{figure}
This subsection justifies \cref{step:rescaling} of \cref{sec:algorithm_structure}.
Consider a rescaled version of the operator just discussed: for $\eta > 0$, let $A_{3,\eta} = \eta((a^\dagger)^3 + a^3)$.
Its deficiency indices are $n_\pm = 3$ for every $\eta>0$.
\Cref{fig:radia-sq3-prefactor} shows the eigenvalues of the associated radius matrices $R_n^+$ for $\eta = 10^{-2}$:
the radii $\rho^j_n$ are clearly much smaller than for $\eta = 1$.
The reason is that, in \cref{eq:generalized-eigenvalue}, every quantity except $z = \iu$ is scaled by $\eta$.
For small $n$ the recurrence is therefore dominated by the spectral parameter, and the $R_n^+$ decay faster at the outset; only asymptotically do the off-diagonal blocks $A_n$ take over, after which all three radii converge to finite values.

The small radii that result are harder to tell from vanishing ones, which is precisely the ambiguity the naive test cannot resolve.
The remedy is to normalise the operator before the iteration begins: we replace $p_M$ by $\lambda p_M$, with $\lambda$ chosen so that the coefficient of the largest photon-number shift equals $1$.
For $A_{3,\eta}$ this gives $\lambda = 1/\eta$ and returns the operator to $A_{3,1}$.
The strategy has proved reliable across the examples we have tested.

\subsection{The Raabe criterion}\label{sec:raabe}
\begin{figure}[t]
    \centering
    \includegraphics[width=0.9\linewidth]{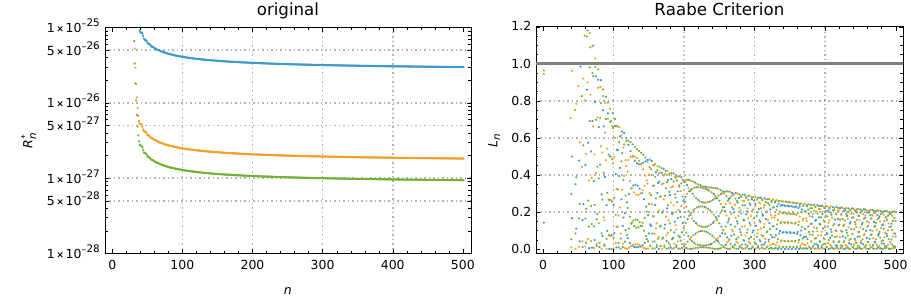}
    \caption{Left: eigenvalues $\rho^j_n$ of the radius operators $R_n^+$ associated with the operator $(a^\dagger)^3 + a^3 + 20 a^\dagger a$. Right: the transformed sequence $L_n^j =  n\left(1-\rho^j_{n+1}/\rho^j_n\right)$.
    While the plot alone does not reveal that the $\rho^j_n$ fail to converge to $0$, the sufficient criterion $L_n^j > 1 + \epsilon_0$ for convergence to $0$ is not met.}
    \label{fig:radia-sq3-bigkerr}
\end{figure}
Together with \cref{sec:convergence}, this subsection justifies \cref{step:decision} of \cref{sec:algorithm_structure}.
In the previous example, we were able to avoid having to deal with small but non-zero limiting radii $\rho^j>0$ by rescaling the operator $A_{3,\eta}$.
The following example illustrates that this is not always possible.
We consider the operator
\begin{equation}
    A_{3,1,\lambda} = (a^\dagger)^3 + a^3 + \lambda a^\dagger a \, ,
\end{equation}
which, by \cite{fischer-selfadjointrealizationshigherorder-2026}, has deficiency indices $n_\pm = 3$ for all $\lambda \in \rnum$, so that all three radius bands must converge to non-zero limiting values.
\Cref{fig:radia-sq3-bigkerr} shows that for $\lambda = 20$ the radii $\rho^j_n$ become so small as to be hard to distinguish from vanishing ones.
We therefore need a more reliable way of deciding whether $\lim_{n \to \infty} \rho^j_n = 0$.
A number of methods for enhancing or accelerating the convergence of sequences can be found in \cite{weniger-nonlinearsequencetransformations-1989,bojanic-unifiedtheoryregularly-1973}.
In our case, a different approach proved useful.
We consider the series $\sum_{n= 0}^\infty \rho^j_n$.
If the series converges, it follows directly that $\rho^j = \lim_{n \to \infty} \rho^j_n = 0$.
A powerful criterion for checking the convergence of a series is the Raabe criterion \cite[p.~285]{knopp-theoryapplicationinfinite-1951}:
if the quantity
\begin{equation}
   L_n^j = n\left(1-\frac{\rho^j_{n+1}}{\rho^j_n}\right) 
\end{equation}
is eventually larger than $1+\epsilon_0$ for some $\epsilon_0>0$, the series converges; if it is eventually smaller than $1$, the series diverges.
The former case thus yields a sufficient criterion for $\lim_{n \to \infty} \rho^j_n = 0$, while the latter is a necessary condition for $\lim_{n \to \infty} \rho^j_n > 0$.
To increase numerical stability, we increase an additional parameter $r > 1$, which, together with its inverse $r^{-1}$, is used to decide whether $L_n^j > 1$ or $L_n^j < 1$ holds eventually. 
We obtain the following decision process for given threshold parameters $\epsilon_0,\epsilon_1$:
\begin{enumerate}[(i)]
    \item if $L_n^j > r$ for all large enough $n$, we conclude $\lim_{n \to \infty} \rho^j_n = 0$;
    \item if $L_n^j < r^{-1}$ and $\rho^j_n > \epsilon_1$ for all large enough $n$, we conclude $\lim_{n \to \infty} \rho^j_n >0$;
    \item if $L_n^j < r^{-1}$ and $\rho^j_n < \epsilon_1$ for all large enough $n$, the algorithm returns no result;
    \item if $L_n^j$ oscillates around $1$, the algorithm returns no result.
\end{enumerate}
\Cref{fig:radia-sq3-bigkerr} illustrates this process.
By checking only the thresholds $\rho^j_n < \epsilon_0$ for all large enough $n$, we would have concluded that $\lim_{n \to \infty} \rho^j_n = 0$, since the $\rho^j_n$ are of order $10^{-26}$.
The sufficient criterion $L_n^j>1+\epsilon_0$ for convergence to $0$ fails, however; instead, $L_n^j < 1$ holds for all large enough $n$, indicating that $\lim_{n \to \infty} \rho^j_n > \epsilon_1$ is possible.
Indeed, inspecting \cref{fig:radia-sq3-bigkerr} closely, we observe that all radii appear to converge to a very small but finite limiting value, in agreement with the analytical results.

The Raabe criterion is therefore mainly useful in avoiding the false conclusion that $\lim_{n \to \infty} \rho^j_n =0$:
if the $\rho^j_n$ are very small but $L_n^j < 1$ holds, the numerical results warrant a closer look.

\subsection{Choosing the sampling size}\label{sec:depth}
\begin{figure}
    \centering
    \includegraphics[width=0.9\linewidth]{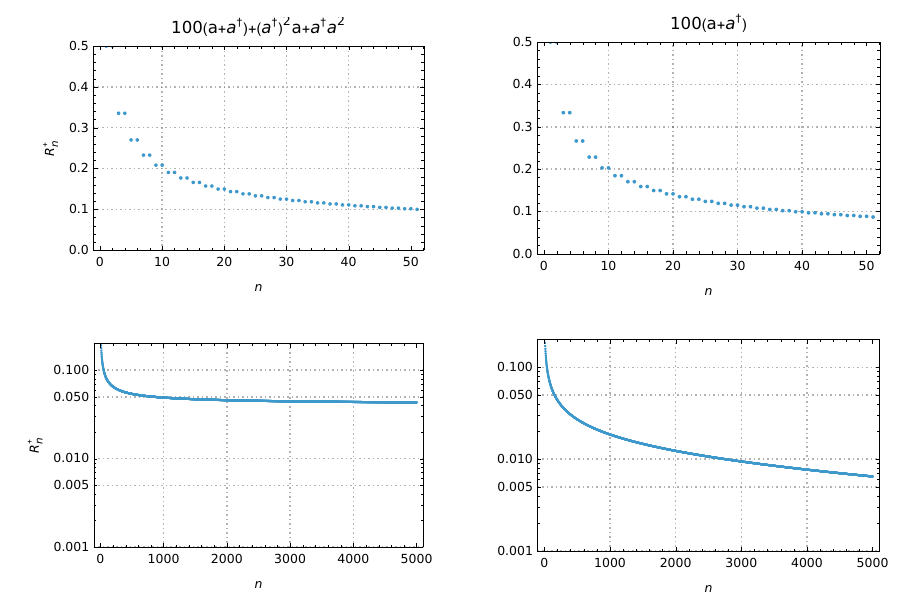}
    \caption{Radii $R_n^+$ associated with the operators $100(a+a^\dagger) + (a^\dagger)^2 a + a^\dagger a^2$ (left) and $100(a^\dagger a)$ (right) for $n \leq 50$ (top) and $n \leq 5000$ (bottom). For small $n$, the term $100(a+a^\dagger)$ dominates the behaviour of the radii, and only for big $n$ the asymptotic behaviour of the coefficients becomes visible.}
    \label{fig:radia-sq21-length}
\end{figure}

This subsection justifies \cref{step:depth} of \cref{sec:algorithm_structure}.
To illustrate the next shortcoming of the naive approach, we consider the polynomial
\begin{equation}
    B_\eta= \eta(a+a^\dagger) + (a^\dagger)^2 a + a^\dagger a^2 \, .
\end{equation}
By \cref{prop:decomposition}, $B_\eta$ is unitarily equivalent to a block Jacobi operator with $d=1$---that is, to a scalar Jacobi operator---with off-diagonal entries
\begin{equation}
    a_n = \eta\sqrt{n+1} + \sqrt{(n+1) n^2} \sim n^{3/2} \, ,
\end{equation}
and, by \cite{fischer-essentiallysingularlimits-2026}, the deficiency indices of $B_\eta$ are $(1,1)$ for all $\eta \in \rnum$.
It follows that the radii $R^\pm_n$ converge to finite numbers $R^\pm > 0$ as $n \to \infty$.

For large $\eta$ and small $n$, however, the term $\eta(a+a^\dagger)$ dominates, and $a_n$ scales like $\eta n^{1/2}$.
Over that range the radii of $B_\eta$ are indistinguishable from those of $C_\eta = \eta(a+a^\dagger)$, as \cref{fig:radia-sq21-length} shows, and the Raabe criterion of \cref{sec:raabe} does no better.
Since $C_\eta$ has deficiency indices $(0,0)$ while $B_\eta$ does not, stopping the iteration at such small $n$ would return the wrong verdict for
$B_\eta$.

Once $n$ is large enough for $\sqrt{(n+1)n^2}$ to dominate $\eta\sqrt{n+1}$, the radii of $B_\eta$ decay more slowly and settle at $R^\pm>0$, while those of $C_\eta$ keep decreasing, and the correct results $(1,1)$ and $(0,0)$ are obtained.
The sampling size $N$ must therefore be large enough for the coefficients to have reached their asymptotic form---here $n \gg \eta$---before the radii are interpreted at all.
The implementation of the algorithm automatically chooses $N$ large enough that all coefficients
are close to their asymptotic expansions.

\subsection{Gaussian operators}\label{sec:carleman-ops}
\begin{figure}
    \centering
    \includegraphics[width=0.9\linewidth]{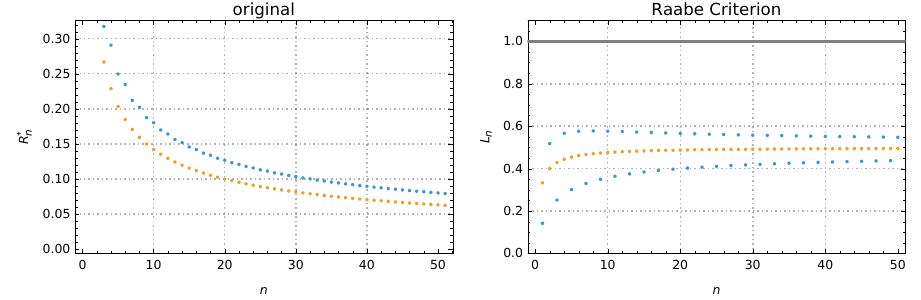}
    \caption{Eigenvalues $\rho^j_n$ of the radius operators $R_n^+$ associated with the squeezing operator $(a^\dagger)^2 + a^2$ (left) and the transformed sequence $L_n^j = n\left(1-\rho^j_{n+1}/\rho^j_n\right)$ (right). 
    As both bands converge to $0$ subalgebraically, the Raabe criterion fails to detect this slow convergence.}
    \label{fig:sq2plot}
\end{figure}
This subsection justifies \cref{step:shortcut} of \cref{sec:algorithm_structure}.
Gaussian operators, which are known to have deficiency indices $(0,0)$ by the Carleman criterion (cf. \Cref{prop:carleman}), are best treated separately.
We consider the usual second-order squeezing operator $A_2 = (a^\dagger)^2 + a^2$, and plot the eigenvalues of the associated radius matrices $R_n^+$ in \cref{fig:sq2plot}.
Since the deficiency indices of $A_2$ are $(0,0)$, the radii must decay to $0$; \Cref{fig:sq2plot} shows, however, that the decay is very slow.
Indeed, one can prove without difficulty that the radii $\rho^j_n$ decay as $O(\log(n)^{-1})$, which in particular implies that the series $\sum_{n = 0}^\infty \rho^j_n$ diverges.
The Raabe criterion therefore fails to detect this slow convergence to $0$, as $L_n^j < 1$ for all $n$ shown in \cref{fig:sq2plot}.
In particular, this demonstrates that $L_n^j<1$ is a necessary but not a sufficient condition for $\lim_{n \to \infty} \rho^j_n > 0$.
Detecting such slow convergence of the $\rho^j_n$ to $0$ is therefore delicate.
In the case of the second-order squeezing operator $A_2$, there is a simple remedy:
we can simply use the analytical result for the deficiency indices.

Numerical evidence suggests that this problem of sub-algebraic (here, $(\log n)^{-1}$) convergence to $0$ arises only for Gaussian operators, which are known always to have deficiency indices $(0,0)$.
In particular, 
across all the
higher-order polynomials examined in this work, the radii either converge to a
non-zero limit or decay algebraically, and in neither case does the criterion of
\cref{sec:raabe} fail.
We therefore conjecture that the phenomenon is specific to degree two.
This costs nothing in practice: Gaussian operators are recognized from $M$
before the iteration begins, and their deficiency indices are known to vanish
by \cref{prop:carleman}, so the algorithm never reaches the regime in which the
criterion would be unreliable.

\subsection{Using the package}\label{sec:implementation}

\begin{table}[t]
    \centering
    \begin{tabular}{l l}
        \toprule
         Function & Returns  \\
         \toprule
         \verb|QuickDeficiencyIndices[poly_]| & Deficiency indices of \verb|poly| with user-friendly output\\
         \verb|DeficiencyPlot[poly_,z_,n_]| & Radii associated with \verb|poly| up to \verb|n|\\
         \verb|RadiusOperators[poly_,z_,n_]| & Radius matrices associated with \verb|poly| up to \verb|n| \\
         \verb|AdaptiveDeficiencyIndices[poly_]| & Deficiency indices of \verb|poly| with metadata\\
    \end{tabular}
    \caption{
    Most important functions in the Mathematica package~\cite{notebook}.
    In all cases, \texttt{poly\_} is a polynomial in $a$ and $a^\dagger$, expressed as a  non-commutative product in the variables \texttt{a} and \texttt{adag} or \texttt{x},\texttt{q} and \texttt{p}, \texttt{z\_} a complex number, and \texttt{n\_} an integer.
    }
    \label{tab:functions}
\end{table}
We briefly comment on the implementation of the algorithm in Mathematica~\cite{Mathematica} and on some sensible default values.
The package, together with a notebook reproducing every example and figure of this work, is available at~\cite{notebook}; the relevant functions are summarized in \cref{tab:functions}.

The package uses Mathematica's \verb|WeylAlgebra| to implement the bosonic algebra, denoting $a$ and $a^\dagger$ by \verb|a| and \verb|adag|.
The relevant modules accept input operator polynomials in the annihilation and creation operators \verb|a| and \verb|adag| as well as in momentum \verb|p| and position \verb|x| (or alternatively \verb|q|), where Mathematica's  \verb|NonCommutativeMultiply| acts as the product.

We comment briefly on the modules most important to the user, as listed in \cref{tab:functions}.
The package also includes some quality of life functions for working with bosonic operators and block Jacobi operators.

\verb|RadiusOperators[poly_,z_, n_]| with $z = \pm \iu$ returns the radius operators $R_l^\pm$ associated with \verb|poly| for all $l \leq n$, and \verb|DeficiencyPlot[poly_,z_, n_]| returns their eigenvalues directly, ready to be plotted.
Both functions have the following options: \verb|Precision| lets the user adjust the working floating point precision, in particular it can be increased whenever Mathematica emits warnings on lacking stability of matrix inversion. 
Moreover, using the option \verb|StartIndex| one can ignore the first (potentially non-invertible) block $A_n$ of the block Jacobi operator associated with \verb|poly|, as described in \cref{lem:full-decomposition}.
Finally, the two functions warn the user whenever the rescaling of the operator (cf. \Cref{sec:rescaling}) or rescaling of the algorithm depth (cf. \Cref{sec:depth}) is recommended, and the rescaling can be automatically applied by setting \verb|AutoScale| to \verb|True|.
These warnings can be disable using \verb|Warnings -> False|.
A sensible initial sampling size for visually inspecting convergence is $n=50$.

The most important module is \verb|AdaptiveDeficiencyIndices[poly_]|, which calculates the deficiency indices of \verb|poly| automatically, returning them together with metadata about the convergence criteria.
It follows the convergence tests outlined in \cref{sec:algorithm_structure}.
Since numerical calculations are inherently imprecise, we employ a number of tests to detect when such imprecision significantly affects the results:
\begin{enumerate}[(i)]
    \item Mathematica recognizes badly conditioned matrices while calculating the $R_n^\pm$;
    \item the eigenvalues $\rho^j_n$ have a non-negligible imaginary part;
    \item the eigenvalues $\rho^j_n$ are not monotonically decreasing.
\end{enumerate}
If at least one of these problems occurs in an iteration, the internal precision of Mathematica is increased, which in particular requires the input coefficient matrix $M$ to be given with sufficient precision.
Ideally, the coefficients of the input polynomial \verb|poly| are therefore given as exact numbers and not in floating point representation, so that Mathematica can evaluate them with arbitrarily high precision.
Since the algorithm need not terminate (for instance if $\epsilon_1 > \lim_{n \to \infty} \rho^j_n > 0$), the implementation imposes a maximum number of iterations and a maximum working precision.
If this maximum is reached and the algorithm has still not returned a definite result, the user is advised to inspect the plots, such as those in \cref{fig:radia-sq3-bigkerr}, and to adjust the hyperparameters accordingly.
The most important options of  \verb|AdaptiveDeficiencyIndices[poly_]| are the following:
\verb|Epsilon0|, \verb|Epsilon1| and \verb|RaabeRatio| are the thresholds $\epsilon_0$, $\epsilon_1$ and $r$ from \cref{sec:algorithm_structure}.
The default values $\epsilon_0 = 10^{-3}$, $\epsilon_1 = 10^{-2}$ and $r=2$ proved to be reliable choices. 
Moreover, \verb|ConvergenceTest| is the number $N_c$ of trailing eigenvalues on which the convergence tests are applied, with a default value of $25$.
Finally, \verb|DebugPrint->True| allows the user to inspect the convergence status of each iteration of the algorithm. 
The error codes returned by \verb|AdaptiveDeficiencyIndices[poly_]| are documented in the module. 

Finally, \verb|QuickDeficiencyIndices[poly_]| is a wrapper for \verb|AdaptiveDeficiencyIndices[poly_]| which provides simplified, color-coded output, and the module which is recommended for the end user.
We refer to the notebook~\cite{notebook} for detailed examples.

\end{document}